\documentclass[sigconf, nonacm]{acmart}
\usepackage{array}             
\usepackage{algorithm2e}
\usepackage{float}
\usepackage{rotating}
\usepackage{appendix}
\usepackage{listings}
\usepackage[english]{babel}
\usepackage{mathtools}         
\usepackage{enumitem}
\usepackage{bm}
\usepackage{amsthm}

\newtheorem{theorem}{Theorem}[section]
\newtheorem{lemma}[theorem]{Lemma}

\setlist{nosep}

\title{Inverse Quadratic Decay in Random Subset Sum}

\author{Edwin Chen}
\affiliation{
  \institution{Dept. of Electrical and Computer Engineering, Portland State University}
  \city{Portland}
  \state{OR}
  \country{USA}
}
\email{echen1ffa@gmail.com}

\author{Christof Teuscher}
\affiliation{
  \institution{Dept. of Electrical and Computer Engineering, Portland State University}
  \city{Portland}
  \state{OR}
  \country{USA}
}
\email{teuscher@pdx.edu}



\begin{document}

\begin{abstract}
The Subset Sum Problem is a fundamental NP-complete problem in cryptography and combinatorial optimization, with many real-world applications. The Random Subset Sum Problem (RSSP) is a more applicable version of subset sum, where numbers are drawn from some i.i.d input distribution. We present an algorithm that, with probability $1-\delta$, constructs the same $O(B/w)$ mesh as Da Cunha et al.~\cite{DaCunhaRSSPMesh}, while trimming to $w$ elements throughout and running in $O(w\log w)$ time. Then, we present a novel beam search heuristic running in linearithmic time w.r.t list size $n$ and beam width $w$ using the mesh that gives an expected error of $O\!\left(\frac{B}{nw^2}\right)$ under a standard mean-field assumption with equal standard deviation, demonstrating the practical effectiveness of meshing to achieve error decay. The algorithm is empirically robust to multiple input distributions and can naturally extend to variants with simple changes to the scoring heuristic, establishing a new practical baseline for robust subset sum error decay and $\epsilon$-approximation theory.
\end{abstract}

\maketitle
 
\section{Introduction}

The Subset Sum problem is a classic NP-complete problem \cite{garey_computers_1979} with applications in resource allocation \cite{abboud2019andsubsetsum}, cryptography \cite{bonnetain_improved_2020}, financial auditing \cite{biesner_subset_2022}, and combinatorial optimization. Given a multiset \(S \subset \mathbb{Z}\) and a target \(T \in \mathbb{Z}\), the goal is to find a subset \(V \subseteq S\) whose sum is as close as possible to \(T\):
\begin{equation}
\text{Answer}(S, T) \;=\; \min_{V \subseteq S} \left| \sum_{v \in V} v - T \right|.
\end{equation}
We study the Random Subset Sum Problem (RSSP), where \(n\) i.i.d.\ weights are drawn from a fixed distribution, typically \(U(1,B)\), focusing on expected error decay for RSSP, which is useful well beyond cryptographic settings. 

We present and evaluate a simple beam search heuristic and prove:
\begin{itemize}[leftmargin=*]
\item \textbf{Explicit-Constant meshing bound:}
  Phase~A fills all $w$ buckets with probability at least $1-\delta$ in
  $7.96\,\log_2 w + 5.19\,\log_2(1/\delta)+O(1)$ iterations, implying an
  $O(B/w)$ mesh constructed in $O(w\log w)$ time with width-$w$ trimming at
  every step.
\item \textbf{MITM beam:} \(\mathbb{E}[\text{error}] = O\!\bigl(B/(n\,w^2)\bigr)\) with asymptotically equal standard deviation.
\item \textbf{Complexity:} \(O(nw\log w)\) time and \(O(w)\) memory for search; exact reconstruction in \(O(nw)\) time with \(O(w\sqrt{n})\) memory.
\item \textbf{Robustness:} The guarantees hold, up to constants, for a broad class of i.i.d. input distributions; empirical results match the theory.
\end{itemize}

To our knowledge, prior work has not emphasized expected error-decay rates for RSSP; many papers focus on uniform inputs (Sections \ref{sec:rep}, \ref{sec:attacks}) or worst-case settings (Section \ref{sec:exact}). For the methods that have a similar scope to our heuristic (Section \ref{sec:heuristics}), our analysis and experiments suggest the proposed heuristic is a practical baseline for robust Subset Sum approximation.

\section{Related Works}

\subsection{Exact and Approximate Algorithms for the Deterministic Version}
\label{sec:exact}
When all elements of \(S\) are nonnegative, both exact and approximation schemes are well studied. The classic dynamic program runs in pseudo-polynomial time \(\widetilde{O}(n\,w_{\max})\), where \(w_{\max}=\max S\). Chen et al.~\cite{chen_faster_2023} improved this to \(O(n + w_{\max}^{3/2})\). Fully Polynomial-Time Approximation Schemes (FPTAS) such as Gens and Levner~\cite{gens_approximation_1979} and Chen et al.~\cite{chen_partition_2024} offer guarantees with near-optimal dependence on \(\varepsilon\), the relative error. Table~\ref{tab:exact} summarizes representative results we compare against.

\begin{table}[htbp]
\centering
\caption{Exact and approximate algorithms for Subset Sum (nonnegative case).}
\label{tab:exact}
\begin{tabular}{|>{\raggedright\arraybackslash}p{0.18\linewidth}|>{\raggedright\arraybackslash}p{0.18\linewidth}|>{\raggedright\arraybackslash}p{0.15\linewidth}|>{\raggedright\arraybackslash}p{0.16\linewidth}|>{\raggedright\arraybackslash}p{0.18\linewidth}|}
\hline
\textbf{Type} & \textbf{Reference} & \textbf{Time} & \textbf{Space} & \textbf{Notes} \\
\hline
Exact (DP)               & Classical DP (Bellman '56) & $\widetilde{O}(n\,w_{\max})$ & $O(w_{\max})$ & Pseudo-polynomial; $w_{\max}=\max S$ \\
\hline
Exact (DP, randomized)   & Bringmann '17              & $\widetilde{O}(n + T)$       & $O(T)$        & Near-linear randomized DP \\
\hline
Exact (DP, deterministic)& Koiliaris \& Xu '17        & $\widetilde{O}(\sqrt{n}\,T)$ or $\widetilde{O}(T^{4/3})$ & $O(T)$ & Subquadratic DP \\
\hline
FPTAS (trimming)         & Kellerer \& Mansini        & $\widetilde{O}(n + \varepsilon^{-2})$  
& $O(1/\varepsilon)$ & Classic list trimming \\
\hline
FPTAS (near-linear)      & Chen et al. '24            & $\widetilde{O}(n + 1/\varepsilon)$ & $O(n + 1/\varepsilon)$ & Weak scheme variant \\
\hline
FPTAS (trimming)         & Gens \& Levner '79         & $O(n/\varepsilon)$           & $O(1/\varepsilon)$ & Empirically Fast Baseline \\
\hline
Exact (offset DP)        & Offset-based DP            & $\widetilde{O}(n\,R)$        & $O(R)$        & $R=\sum_{x\in S}|x|$; handles negatives \\
\hline
Probabilistic & \textbf{Ours}
& $O(n w \log w)$
& $O(w)$ (or $O(w\sqrt{n})$ with reconstruction)
& Expected error $O(B/(n w^2))$
\\
\hline
\end{tabular}
\end{table}

\subsection{Heuristic Methods}
\label{sec:heuristics}
Heuristics are attractive for large or mixed-sign instances where exact or FPTAS methods can be slow or inapplicable. Representative approaches include Genetic Algorithms~\cite{nguyen_genetic_2004, goldberg_genetic_1989}, Simulated Annealing~\cite{kirkpatrick_optimization_1983}, Particle Swarm Optimization (PSO)~\cite{kennedy_particle_1995}, Tabu Search~\cite{glover_tabu_1989, glover_tabu_1990, glover_future_1986}, and the Arithmetic Optimization Algorithm (AOA)~\cite{madugula_aoa_2022}. Table~\ref{tab:heuristic} summarizes typical characteristics.

\begin{table}[htbp]
\centering
\caption{Heuristic algorithms for Subset Sum.}
\label{tab:heuristic}
\begin{tabular}{|>{\raggedright\arraybackslash}p{0.25\linewidth}|>{\raggedright\arraybackslash}p{0.24\linewidth}|>{\raggedright\arraybackslash}p{0.20\linewidth}|>{\raggedright\arraybackslash}p{0.20\linewidth}|}
\hline
\textbf{Method} & \textbf{Typical Runtime} & \textbf{Memory} & \textbf{Notes} \\
\hline
Genetic Algorithm & $O(n\cdot gen\cdot pop)$ & $O(n\cdot pop)$ & Evolutionary, tunable \\
\hline
Simulated Annealing & $O(n\cdot \log T)$ & $O(n)$ & Escapes local minima \\
\hline
Particle Swarm Optimization (PSO) & $O(n\cdot particles\cdot iter)$ & $O(n\cdot particles)$ & Swarm-based \\
\hline
Arithmetic Optimization Algorithm (AOA) & $O(n\cdot pop\cdot iter)$ & $O(n\cdot pop)$ & Good for high density \\
\hline
Tabu Search & $O(n\cdot iter)$ & $O(n\cdot tabu)$ & Memory-guided local search \\
\hline
\textbf{Hyper-heuristics} & \emph{Varies} & \emph{Varies} & Combines multiple heuristics \\
\hline
\textbf{Beam Search (Ours)} & $O(n\cdot w\log w)$ & $O(w)$ (or $O(\sqrt{n}\,w)$ with reconstruction) & Deterministic, \textbf{proven error and variance decay} \\
\hline
\end{tabular}
\end{table}

\subsection{Representation Techniques for Subset Sum}
\label{sec:rep}

\begin{table}[htbp]
\centering
\caption{Representation-based algorithms for Subset Sum (average case).}
\label{tab:representation-techniques}
\begin{tabular}{|>{\raggedright\arraybackslash}p{0.23\linewidth}|>{\raggedright\arraybackslash}p{0.20\linewidth}|>{\raggedright\arraybackslash}p{0.23\linewidth}|>{\raggedright\arraybackslash}p{0.22\linewidth}|}
\hline
\textbf{Technique} & \textbf{Alphabet / Params} & \textbf{Complexity (time / space)} & \textbf{Model \& Reference} \\
\hline
Baseline MITM / Schroeppel--Shamir &
$A=\{0,1\}$; 2-way split; hashing/merge &
$\tilde O(2^{n/2})$ time; $\tilde O(2^{n/4})$ space &
Worst case; \cite{horowitz_sahni_1974,schroeppel_shamir_1981} \\
\hline
Representation + dissection (HGJ) &
$A=\{-1,0,1\}$; modular balancing; sparse sampling &
$\tilde O(2^{0.337\,n})$ time; $\tilde O(2^{0.337\,n})$ space &
Uniform Random; \cite{howgravegraham_joux_2010} \\
\hline
Optimized representation (BCJ) &
$A=\{-1,0,1\}$; multi-level dissection; tuned moduli &
$\tilde O(2^{0.291\,n})$ time; $\tilde O(2^{0.291\,n})$ space &
Uniform Random; \cite{becker_coron_joux_2011} \\
\hline
Larger alphabet (classical) &
$A=\{-1,0,1,2\}$ &
$\tilde O(2^{0.283\,n})$ time; $\tilde O(2^{0.283\,n})$ space &
Uniform Random; \cite{bonnetain_improved_2020} \\
\hline
Larger alphabet + Grover (quantum) &
$A=\{-1,0,1,2\}$ &
$\tilde O(2^{0.236\,n})$ time; quasilinear space &
Uniform Random (quantum); \cite{bonnetain_improved_2020} \\
\hline
Larger alphabet + quantum walk &
$A=\{-1,0,1,2\}$ &
$\tilde O(2^{0.216\,n})$ time; quasilinear space &
Uniform Random (quantum); \cite{bonnetain_improved_2020} \\
\hline
\end{tabular}
\end{table}

For uniform random inputs, a powerful line of work accelerates meet-in-the-middle by allowing redundant encodings of a \(0/1\) solution using a larger coefficient alphabet (e.g., \(A=\{-1,0,1\}\) or \(A=\{-1,0,1,2\}\)), as summarized in Table~\ref{tab:representation-techniques}. 

However, there is comparatively little discussion of how such representation-based methods behave when viewed as approximators rather than exact algorithms. While representation-based algorithms implicitly discard large portions of the search space, their analyses typically focus on success probability and running time, rather than on how the induced approximation error decays as a function of retained state size. Because representation techniques filter partial sums using exact modular congruences, they do not naturally preserve 'close' approximations, making their direct adaptation into an expected-error heuristic non-trivial. As a result, the relationship between these techniques and expected-error frameworks has not been systematically analyzed. We view these representation-based techniques as orthogonal to our expected-error framework, rather than directly competing approaches.

\subsection{Attacks on Randomized Subset Sum}
\label{sec:attacks}
RSSP admits probabilistic attacks in certain density regimes. Let the (base-2) density be \(d := n/\log_2 B\). Some ranges are believed “easy,” though algorithms that solve these "easy" ranges have high polynomial degree and/or rely on Shortest Vector Problem (SVP) oracles. The methods of \cite{lagarias_solving_1985} and \cite{coster_improved_1992} rely on SVP oracles. 

\begin{table}[htbp]
\centering
\caption{Density ranges where random Subset Sum is expected polynomial-time solvable.}
\label{tab:attacks-on-randomized}
\begin{tabular}{|>{\raggedright\arraybackslash}p{0.25\linewidth}|>{\raggedright\arraybackslash}p{0.20\linewidth}|>{\raggedright\arraybackslash}p{0.20\linewidth}|>{\raggedright\arraybackslash}p{0.25\linewidth}|}
\hline
\textbf{Density regime} & \(\boldsymbol{\log_2 B}\) & \textbf{Algorithm} & \textbf{Reference} \\
\hline
Very low: \(d < 0.6463\) & \(\log_2 B > 1.546\,n\) & LLL-based lattice reduction & \cite{lagarias_solving_1985} \\
\hline
Low--Moderate: \(0.6463 \le d < 0.94\) & \(1.064\,n < \log_2 B \le 1.546\,n\) & Improved lattice embeddings & \cite{coster_improved_1992} \\
\hline
Quasi-polylog: \(d = \Theta\!\bigl(n/(\log n)^2\bigr)\) & \(\log_2 B = O((\log n)^2)\) & Recursive bit-peeling & \cite{flaxman_solving_2005} \\
\hline
Very high: \(d \ge \Theta\!\bigl(n/\log n\bigr)\) & \(\log_2 B = O(\log n)\) & Classic pseudo-polynomial DP & --- \\
\hline
\end{tabular}
\end{table}

Large density ranges remain challenging in practice, as shown in Table \ref{tab:attacks-on-randomized}.

\section{General Beam Search}

Beam Search is a heuristic graph traversal algorithm that strikes a balance between breadth-first search and greedy best-first search. It has been used extensively in NLP tasks \cite{wiseman_sequence--sequence_2016} \cite{huang_when_2017}. At each step, instead of expanding all possible children (as in breadth-first search), it maintains only a fixed number $w$ of the most promising candidates—known as the \textit{beam width}.

The algorithm proceeds in levels: at each level, every current candidate is expanded by generating its successors. All successors are scored using a heuristic function (e.g., cost, distance, or error), and only the top $w$ are retained for the next level. This pruning mechanism reduces memory and computation compared to exhaustive methods, while preserving some diversity in the search.

Beam Search is widely used in tasks such as sequence decoding in natural language processing (e.g., machine translation), planning, combinatorial optimization, and search in large discrete spaces.

\subsection*{Pseudocode for Beam Search}
\begin{algorithm}[htbp]
  \caption{Generic Beam Search}
  \label{alg:beam-search-generic}
  \KwIn{Initial state $s_0$; beam width $w$; successor function $\texttt{Succ}(s)$; scoring function $\texttt{Score}(s)$; termination condition $\texttt{Done}(\mathcal{W})$}
  \KwOut{Best state found}

  $\mathcal{W} \leftarrow [\,s_0\,]$\tcp*{Beam: current candidates}
  $\texttt{best} \leftarrow s_0$\;

  \While{$\mathcal{W}\neq []$ \textbf{and} not $\texttt{Done}(\mathcal{W})$}{
    $\mathcal{C} \leftarrow []$\tcp*{All successors}
    \ForEach{$s \in \mathcal{W}$}{
      \ForEach{$s' \in \texttt{Succ}(s)$}{
        append $s'$ to $\mathcal{C}$\;
      }
    }
    \If{$\mathcal{C} = []$}{
      \textbf{break}\tcp*{No successors; dead end}
    }
    sort $\mathcal{C}$ in descending order by $\texttt{Score}(\cdot)$\;
    $\mathcal{W} \leftarrow$ first $\min(w,|\mathcal{C}|)$ elements of $\mathcal{C}$\tcp*{Keep top candidates}
    \If{$\texttt{Score}(\mathcal{W}[1]) > \texttt{Score}(\texttt{best})$}{
      $\texttt{best} \leftarrow \mathcal{W}[1]$\;
    }
  }
  \Return{$\texttt{best}$}\;
\end{algorithm}

\noindent
In this formulation:
\begin{itemize}
    \item \texttt{Succ(node)} generates all immediate successors of \texttt{node}.
    \item \texttt{Score(node)} evaluates how close a node is to the goal or optimality.
    \item \texttt{Done} returns \texttt{true} if the final layer of the search DAG has been reached.
    \item $\mathcal{W}$ contains at most $w$ candidates.
\end{itemize}

Beam Search is not guaranteed to find the optimal solution, but often finds high-quality approximations quickly. Its performance depends heavily on the choice of beam width and scoring heuristic \cite{zhou_beam_2020}.

\subsection{Generalizing Beam Search for Subset Sum}

To generalize Beam Search for the Subset Sum Problem, we treat the problem as a sequence of decisions: at each index in the input list, we can either include the element in the subset or exclude it. Each partial solution is a path in a binary tree, where each node represents a partial subset and maintains a running sum.

The algorithm proceeds iteratively, maintaining a beam of the $w$ most promising partial solutions at each step. The score for each partial solution is determined by its absolute deviation from the target. At each step, the beam is expanded by including or excluding the next element in the list, and the top $w$ branches (by smallest error) are retained.

\begin{algorithm}[htbp]
\caption{Beam Search for Closest Subset Sum}
\label{alg:beam-subset-sum}
\KwIn{Set of integers $S = \{s_1, \dots, s_n\}$; target $T$; beam width $w$}
\KwOut{Subset sum closest to $T$}

$\mathcal{W} \leftarrow \{0\}$  

\For{$i \leftarrow 1$ \KwTo $n$}{
    $\mathcal{W} \leftarrow \mathcal{W} \,\cup\, \{x + s_i \mid x \in \mathcal{W}\}$ \\
    Truncate $\mathcal{W}$ to its $w$ closest elements
}

\Return{$\arg\min_{x \in \mathcal{W}} |x - T|$}
\end{algorithm}

This approach requires $O(w)$ space and $O(n \cdot w \log w)$ time per run, because truncating $\mathcal{W}$ to its $w$ closest elements requires sorting. In practice, Beam Search provides strong empirical performance, often outperforming more complex metaheuristics in both accuracy and runtime on randomly generated subset sum instances.

\section{Optimizations for Beam Search}

\subsection{Reconstructing the Optimal Subset}

Although not the main focus of this paper, reconstructing the actual
subset corresponding to the best sum is important in practice. A naive
implementation keeps full parent pointers for all $n$ layers, using
 $O(nw)$ memory, or recomputes layers one by one, taking $O(n^2w)$ time.

\begin{lemma}[Invertibility of Beam Iteration]\label{lem:beam-inv}
Given the beam state \(W_i \in \mathbb{R}^w\) and the step size \(s_i \in \mathbb{R}\), recovering the previous beam \(W_{i-1}\) requires at least \(\Omega(w)\) bits of information.
\end{lemma}

\begin{proof}[Proof]
Given \(W_i\) and \(s_i\), consider all possible predecessor candidates
\[
W_{i-1} = W_i - s_i z_i, \quad z \in \{0,1\}^w.
\]
A candidate \(W_{i-1}\) is valid if
\[
(W_{i-1} + s_i) \cup W_{i-1} = W_i.
\]
Without loss of generality, assume the target value \(T = 0\).  
Then, for each coordinate \(1 \le j \le w\):
\[
\text{if } z_j = 1, \quad |(W_{i-1})_j + s_i| \le |(W_{i-1})_j|, 
\quad \text{else } |(W_{i-1})_j| \le |(W_{i-1})_j + s_i|.
\]
In words, applying the inverse operation \(-z_j s_i\) should preserve the property that the chosen element remains the one closest to the target.  

If \(4\max_j |(W_i)_j| \le s_i\), then all \(2^w\) candidates \(W_{i-1} = W_i - s_i z\) satisfy the above condition, implying that \(W_{i-1}\) cannot be uniquely determined from \(W_i\). Hence, reconstructing \(W_{i-1}\) requires at least \(\Omega(w)\) bits to specify which of the \(2^w\) possible configurations is valid. 
\end{proof}

\paragraph{Checkpointing idea.}  

This motivates the following checkpointing approach: Every $m$ steps we store the current beam. During
reconstruction, we backtrack from the final state to the nearest
checkpoint, then locally recompute at most $m$ layers. This costs
$O(w\cdot n/m)$ memory for the checkpoints and $O(wm)$ memory for the
temporary recomputation. Choosing $m=\sqrt{n}$ balances the two terms,
giving $O(w\sqrt{n})$ memory and $O(nw)$ time overall.

\begin{algorithm}[htbp]
\caption{Beam Search Forward with Checkpoints}
\label{alg:beam-forward-checkpoints}
\KwIn{Integers $S=\{s_1,\dots,s_n\}$; target $T$; beam width $w$}
\KwOut{Best sum $x^*$ and beam checkpoints}
$m \leftarrow \lfloor \sqrt{n}\rfloor$; beam $\leftarrow \{0\}$ \\
checkpoints $\leftarrow \{0\mapsto$ beam$\}$ \\
\For{$i=1$ \KwTo $n$}{
  Expand beam with $s_i$ and trim to width $w$ \\
  \If{$i\bmod m=0$ or $i=n$}{checkpoints[$i$] $\leftarrow$ copy of beam}
}
\Return{best $x^*\in$ beam, checkpoints}
\end{algorithm}

\begin{algorithm}[htbp]
\caption{Reconstruct Subset Using Checkpoints}
\label{alg:beam-reconstruct-checkpoints}
\KwIn{$S,T,w$, best sum $x^*$, checkpoints}
\KwOut{Indices of elements forming subset}
$m \leftarrow \lfloor \sqrt{n}\rfloor$; curr\_sum $\leftarrow x^*$; sol $\leftarrow \emptyset$ \\
\For{$j=n$ down to $1$ in steps of $m$}{
  start $\leftarrow \max(0,j-m)$ \\
  beam $\leftarrow$ checkpoints[start] \\
  Recompute layers $S_{start+1},\dots,S_j$ with parent pointers \\
  Backtrack from curr\_sum to start, appending chosen indices to sol \\
  Update curr\_sum to parent sum at start
}
Reverse sol and return
\end{algorithm}

\section{MITM Beam Search}

A natural extension of the beam heuristic is to split the input into two
halves and run a meet-in-the-middle (MITM) beam search. The left
half generates a small set of ``anchors,'' while the right half searches
against multiple residual targets. This yields a quadratic improvement
in error decay.

The following analysis uses the Uniform distribution $U(0, B)$, but the algorithm is empirically robust to many input distributions.

\subsection{Input Transformation}
SSP can be modeled as a game of take/not take, corresponding to 1/0. We can use this observation to formulate the following transformation: pick a random subset of elements $\pi$ using a Bernoulli distribution. Then $T := T-\sum s_{\pi_i}$ and $\pi_i:=-\pi_i$. For those indices we instead make the decision of not take/take. In this way, WLOG we turn any input distribution into one that is symmetric around 0. 

Note that this removes the issue of sums drifting to the right over time. This means that we are now drawing elements from a symmetric distribution.

Additionally, assume $T$ is positive; this is true WLOG by complement.

\subsection{Assumptions}

In the following analysis, we make the following assumptions:

\begin{enumerate}
    \item All elements in the set $S$ are drawn independently from the uniform distribution $U(-B, B)$.
    \item The bound $B$ is much larger than the beam width $w$, i.e. $B \gg w$.
    \item The elements of $S$ are mutually independent.
    \item The beam width satisfies $w > 0$.
    \item \textbf{Microscopic Offset Decorrelation (Mean-Field Heuristic):} While beam elements and anchors inherently share ancestral paths from the expansion tree, we assume their fine-scale offsets within their respective Voronoi cells act as pairwise independent continuous variables. Specifically, the dense addition of fresh i.i.d. variables at each step sufficiently mixes the least significant digits, preventing the cross-differences $Z_i - Z_j$ from collapsing into degenerate periodic lattices. This relaxation—modeling structurally dependent tree-paths as independent, well-distributed random variables to prevent degenerate state-space collapse—is a standard heuristic in average-case Subset Sum analysis, canonically utilized in the analysis of list-merging and representation techniques~\cite{howgravegraham_joux_2010, becker_coron_joux_2011, wagner_generalized_2002}, and is analogous to the mean-field approximations used to analyze the Number Partitioning phase transition~\cite{mertens_phase_1998, borgs_phase_2001}.
\end{enumerate}

\subsection{Algorithm}

\begin{enumerate}[leftmargin=*]
\item \emph{Phase A (anchors).} Expand the left half,
truncating after each step to one element per bucket of size $B/w$, choosing randomly. The resulting anchor set $A=\{a_1,\dots,a_w\}$ has maximum gap
$O(B/w)$.
\item \emph{Phase B (multi-target).} For each anchor $a_j$ define
residual $r_j=T-a_j$. Before any element is within $B/w$ of an anchor, run a width-$w$ beam on the right half, scoring
sums by distance to the residual set $\mathcal{R}=\{r_1,\dots,r_w\}$. After this, only keep one beam element per anchor region $[a_j-B/2w, a_j+B/2w]$, breaking ties within a region by distance to the representative anchor.
\end{enumerate}

\begin{algorithm*}[htbp]
\caption{MITM Beam Search (bucketed anchors + residual-guided right-half beam)}
\label{alg:mitm-barebones}
\KwIn{Integers $S=\{s_1,\dots,s_n\}$, target $T$, beam width $w$, bound $B$}
\KwOut{Best right-half sum $x^\star$ minimizing distance to residuals}
$n_L \leftarrow \lceil C\log w\rceil$\;
$S_L \leftarrow (s_1,\dots,s_{n_L})$;\quad $S_R \leftarrow (s_{n_L+1},\dots,s_n)$\;

\BlankLine
\textbf{Phase A (anchors from left half via one-per-bucket).}\;
$\Delta \leftarrow B/w$ \tcp*[r]{bucket width on $[-B/2,B/2]$}
$\mathcal{W} \leftarrow \{0\}$\;
\For{$i=1$ \KwTo $n_L$}{
  $\text{expanded} \leftarrow \mathcal{W}\ \cup\ \{x+s_i : x\in \mathcal{W}\}$\;
  $E \leftarrow \{x\in \text{expanded} : -B/2 \le x \le B/2\}$\;
  \tcp{Partition $[-B/2,B/2]$ into $w$ half-open buckets $\mathcal{I}_j$}
  $\mathcal{W} \leftarrow \emptyset$\;
  \For{$j=1$ \KwTo $w$}{
    $\mathcal{I}_j \leftarrow \big[-B/2+(j-1)\Delta,\,-B/2+j\Delta\big)$\;
    $C_j \leftarrow \{x\in E : x\in \mathcal{I}_j\}$\;
    \If{$C_j \neq \emptyset$}{
      pick $x$ uniformly at random from $C_j$\;
      $\mathcal{W} \leftarrow \mathcal{W} \cup \{x\}$\;
    }
  }
}
$A \leftarrow \mathcal{W}$\;
$\mathcal{R} \leftarrow \{T-a : a\in A\}$ \tcp*[r]{residual targets for right half}
Delete even-indexed anchors

\BlankLine
\textbf{Phase B (right half: residual-guided, then one-per-anchor).}\;
$\mathcal{W} \leftarrow \{0\}$\;
\For{$i=n_L+1$ \KwTo $n$}{
  $\text{expanded} \leftarrow \mathcal{W}\ \cup\ \{x+s_i : x\in \mathcal{W}\}$\;

  \tcp{Distance to residual set}
  \ForEach{$x \in \text{expanded}$}{
    $d(x) \leftarrow \min_{r\in \mathcal{R}} |x-r|$\;
  }

  \tcp{Detect whether we have entered the anchor neighborhood regime}
  $H \leftarrow \{x\in \text{expanded} : d(x) \le \Delta\}$\;

  \eIf{$H=\emptyset$}{
    \tcp{Pre-hit regime: keep the $w$ closest sums to residuals}
    $\mathcal{W} \leftarrow$ the $w$ elements of $\text{expanded}$ with smallest $d(x)$\;
  }{
    \tcp{Post-hit regime: one-per-anchor over Voronoi cells of $\mathcal{R}$}
    \ForEach{$x \in \text{expanded}$}{
      $r(x) \leftarrow \arg\min_{r\in \mathcal{R}} |x-r|$ \tcp*[r]{closest-anchor assignment; break ties arbitrarily}
    }
    $\mathcal{W} \leftarrow \emptyset$\;
    \ForEach{$r \in \mathcal{R}$}{
      $C_r \leftarrow \{x\in \text{expanded} : r(x)=r\}$\;
      \If{$C_r \neq \emptyset$}{
        \tcp{Keep the best representative inside anchor $r$'s Voronoi region}
        pick $x_r \in \arg\min_{x\in C_r} |x-r|$ \tcp*[r]{ties arbitrary (or random)}
        $\mathcal{W} \leftarrow \mathcal{W} \cup \{x_r\}$\;
      }
    }
  }

}
$x^\star \leftarrow \arg\min_{x\in \mathcal{W}} d(x)$\;
\Return $x^\star$\;
\end{algorithm*}

\subsection*{Phase A}
Let the set of anchors at step $j$ be
\[
A_j = \{a_1^{(j)}, a_2^{(j)}, \dots, a_{|A_j|}^{(j)}\}.
\]
Let $\Delta := B/w$ and let $C_i$ denote the center of bucket $i$, so that bucket $i$ corresponds to the interval
\[
\mathcal{I}_i \;=\; \Big[C_i - \frac{\Delta}{2},\; C_i + \frac{\Delta}{2}\Big].
\]

At step $j$, bucket $i$ is filled if there exists some $a \in A_j$ such that
\[
s_j \in (\mathcal{I}_i - a)\; \cap\; [-B,B].
\]
Define
\[
I_{i, j} = \bigcup_{a \in A_j} \big( (\mathcal{I}_i - a)\cap [-B,B]\big).
\]
Then the probability that bucket $i$ is filled at step $j$ is
\[
\Pr(\text{bucket $i$ filled at step $j$}) \;=\; \frac{\mathrm{Leb}(I_{i, j})}{2B}.
\]

Considering only even-indexed buckets ensures that none of the intervals $\mathcal{I}_i-a, a\in A_j$ intersect, as each $a \in A$ will have one bucket between them, and hence have a spacing of at least $B/w$. Only considering even-indexed buckets can only underestimate the actual lebesgue measure. Hence, by multiplying the total size of buckets with the fraction of buckets filled, we obtain
\[
\mathrm{Leb}(I_{I, j})\;\;\geq\;\;\Omega\!\left(B \cdot \frac{|A_j|}{w}\right).
\]

Thus, conditioning on bucket $i$ not yet being filled,
\[
\Pr(\text{bucket $i$ filled at step $j$ } \mid \text{ not filled yet}) 
\;\;\geq\;\; \Omega\!\left(\frac{|A_j|}{w}\right).
\]

By linearity of expectation, the expected number of new buckets filled at step $j$ is
\[
\mathbb{E}[\Delta |A_j|] 
\;=\; (w - |A_j|)\cdot \Omega\!\left(\frac{|A_j|}{w}\right).
\]

Therefore, the recurrence for the expected number of anchors is
\[
\mathbb{E}[\,|A_{j+1}|\,]
\;\geq\; |A_j| \;+\; c\,(w - |A_j|)\frac{|A_j|}{w},
\]
for some constant $c>0$. Note that eliminating even anchors reduces $c$ by a factor of $1/2$, and pessimistically considering boundary effects on $U(-B, B)$ yields another factor of $1/2$, which gives $c\geq 1/4$.

This recurrence grows geometrically, so after $O(\log w)$ steps a constant fraction of buckets will be filled.

\subsubsection{Phase A: explicit bucket coverage bounds in $O(\log w)$ steps}
\label{sec:phaseA-logw-coverage}

We now formalize the bucket-filling dynamics of Phase~A and give
\emph{explicit constants} for (i) reaching half coverage in expectation,
(ii) reaching essentially full coverage in expectation, and (iii)
filling all buckets with high probability.

Recall the Phase~A argument: after restricting to even-indexed buckets to
enforce disjointness, the total reachable measure contributed by the
current anchor set $A_j$ is $\Omega(B\cdot |A_j|/w)$. Pessimistically
accounting for boundary truncation to $[-B,B]$ loses another constant
factor. Concretely, we obtain the following bound: conditional on a
bucket $i$ being unfilled at step $j$,
\begin{equation}
\Pr(\text{bucket $i$ is filled at step $j$} \mid \text{$i$ unfilled})
\;\ge\; \frac{1}{4}\cdot \frac{|A_j|}{w}.
\label{eq:phaseA-fill-prob}
\end{equation}

By linearity of expectation, the expected number of newly filled buckets
satisfies
\[
\mathbb{E}\big[\Delta |A_j|\big]
\;\ge\; (w-|A_j|)\cdot \frac{1}{4}\frac{|A_j|}{w}.
\]
Equivalently, for all $j$,
\begin{equation}
\mathbb{E}[\,|A_{j+1}|\,] \;\ge\; |A_j|
\;+\; \frac{1}{4}\,(w-|A_j|)\frac{|A_j|}{w}.
\label{eq:phaseA-anchor-recurrence}
\end{equation}

Let
\[
x_j := \frac{\mathbb{E}[|A_j|]}{w}\in[0,1]
\]
denote the expected fraction of filled buckets. Dividing
\eqref{eq:phaseA-anchor-recurrence} by $w$ yields the discrete logistic
drift inequality
\begin{equation}
x_{j+1}-x_j \;\ge\; \frac{1}{4}\,x_j(1-x_j).
\label{eq:phaseA-logistic}
\end{equation}

\paragraph{Continuous comparison.}
The inequality \eqref{eq:phaseA-logistic} is naturally compared to the
logistic ODE
\begin{equation}
\frac{dx}{dt} = \frac{1}{4}\,x(1-x), \qquad x(0)=x_0,
\label{eq:phaseA-ode}
\end{equation}
whose explicit solution is
\begin{equation}
x(t) = \frac{1}{1+\Big(\frac{1}{x_0}-1\Big)e^{-t/4}}.
\label{eq:phaseA-ode-solution}
\end{equation}
With $x_0=1/w$ (since $|A_0|=1$), this becomes
\begin{equation}
x(t)=\frac{1}{1 + (w-1)e^{-t/4}}.
\label{eq:phaseA-ode-solution-w}
\end{equation}

\begin{lemma}[Half coverage in expectation]
\label{lem:phaseA-half-explicit}
There exists an explicit iteration index
\[
j_{\mathrm{half}} \;\le\; 4\ln(w-1) \;\le\; 4\ln w
\]
such that
\[
\mathbb{E}[|A_{j_{\mathrm{half}}}|] \;\ge\; \frac{w}{2}.
\]
Equivalently, Phase~A reaches half coverage in at most
\[
j_{\mathrm{half}} \;\le\; 4\ln w \;=\; (4\ln 2)\,\log_2 w \;\le\; 2.773\,\log_2 w
\]
steps.
\end{lemma}

\begin{proof}
Setting $x(t)=1/2$ in \eqref{eq:phaseA-ode-solution-w} yields
$(w-1)e^{-t/4}=1$, i.e.\ $t=4\ln(w-1)$. The stated bounds follow from
$\ln(w-1)\le \ln w$ and $\ln w=(\ln 2)\log_2 w$.
\end{proof}

\paragraph{Geometric contraction after half coverage.}
For $x_j\ge 1/2$, define the expected unfilled fraction
$u_j:=1-x_j$. From \eqref{eq:phaseA-logistic},
\[
x_{j+1}-x_j \;\ge\; \frac14\cdot \frac12 u_j \;=\; \frac18 u_j,
\]
which implies
\begin{equation}
u_{j+1} \;\le\; \frac78\,u_j.
\label{eq:phaseA-geometric}
\end{equation}
Thus, once half the buckets are filled, the remaining unfilled mass
contracts geometrically at rate $7/8$ per step.

\begin{lemma}[Essentially full coverage in expectation]
\label{lem:phaseA-all-explicit}
Let $j_{\mathrm{half}}$ be as in Lemma~\ref{lem:phaseA-half-explicit} and
define
\[
t_{\mathrm{exp}} \;:=\; 
\left\lceil \frac{\ln(w/2)}{\ln(8/7)} \right\rceil,
\qquad
j_{\mathrm{all,exp}} \;:=\; j_{\mathrm{half}} + t_{\mathrm{exp}}.
\]
Then the expected number of unfilled buckets is at most one:
\[
\mathbb{E}[\,w-|A_{j_{\mathrm{all,exp}}}|\,] \;\le\; 1.
\]
Moreover, using $\ln(8/7)\approx 0.13353$ (so $1/\ln(8/7)\approx 7.489$),
we obtain the explicit bound
\[
j_{\mathrm{all,exp}}
\;\le\;
4\ln w \;+\; \frac{\ln(w/2)}{\ln(8/7)} \;+\; 1
\;\le\;
\Bigl(4+\frac{1}{\ln(8/7)}\Bigr)\ln w \;+\; 1
\]
\[
\;\le\; 11.489\,\ln w \;+\; 1
\;=\; (11.489\,\ln 2)\,\log_2 w \;+\; 1
\;\le\; 7.96\,\log_2 w \;+\; 1.
\]
\end{lemma}

\begin{proof}
At time $j_{\mathrm{half}}$ we have $u_{j_{\mathrm{half}}}\le 1/2$.
Applying \eqref{eq:phaseA-geometric} for $t$ additional steps yields
\[
\mathbb{E}[w-|A_{j_{\mathrm{half}}+t}|]
\;=\; w\,u_{j_{\mathrm{half}}+t}
\;\le\; \frac{w}{2}\left(\frac78\right)^t.
\]
Choosing $t=t_{\mathrm{exp}}$ makes the right-hand side at most $1$.
The numerical constants follow from $1/\ln(8/7)\approx 7.489$ and
$\ln 2\approx 0.6931$.
\end{proof}

\begin{theorem}[All buckets filled with high probability]
\label{lem:phaseA-all-whp}
Fix $\delta\in(0,1)$. Let
\[
t_\delta \;:=\;
\left\lceil \frac{\ln\!\bigl(w/(2\delta)\bigr)}{\ln(8/7)} \right\rceil,
\qquad
j_\delta \;:=\; j_{\mathrm{half}} + t_\delta.
\]
Then
\[
\Pr\bigl(|A_{j_\delta}|<w\bigr) \;\le\; \delta,
\]
i.e., Phase~A fills all buckets with probability at least $1-\delta$ by
iteration $j_\delta$.

Moreover, using $1/\ln(8/7)\approx 7.489$, we have the explicit bound
\[
j_\delta
\;\le\;
4\ln w \;+\; \frac{\ln\!\bigl(w/(2\delta)\bigr)}{\ln(8/7)} \;+\; 1
\;\le\;
\Bigl(4+\frac{1}{\ln(8/7)}\Bigr)\ln w \;+\; \frac{\ln(1/\delta)}{\ln(8/7)} \;+\; 1
\]
\[
\;\le\;
11.489\,\ln w \;+\; 7.489\,\ln(1/\delta) \;+\; 1
\]
\[
\;=\;
7.96\,\log_2 w \;+\; 5.19\,\log_2(1/\delta) \;+\; 1.
\]
In particular, for any fixed constant $\delta$ (e.g., $\delta=0.01$),
this is $j_\delta \le 11.489\,\ln w + O(1)$.
\end{theorem}

\begin{proof}
By Lemma~\ref{lem:phaseA-all-explicit}, for any $t\ge 0$,
\[
\mathbb{E}[w-|A_{j_{\mathrm{half}}+t}|]
\;\le\;
\frac{w}{2}\left(\frac78\right)^t.
\]
Markov's inequality gives
\[
\Pr(|A_{j_{\mathrm{half}}+t}|<w)
=
\Pr(w-|A_{j_{\mathrm{half}}+t}|\ge 1)
\;\le\;
\mathbb{E}[w-|A_{j_{\mathrm{half}}+t}|]
\;\le\;
\frac{w}{2}\left(\frac78\right)^t.
\]
Choosing $t=t_\delta$ makes the right-hand side at most $\delta$.
The numerical constants follow as above.
\end{proof}

\paragraph{Consequence: mesh size.}
Once all buckets are filled, selecting one anchor per bucket yields a
mesh with maximum gap $O(B/w)$. Deleting every even-indexed bucket ensures
both minimum and maximum gaps are $\Theta(B/w)$.

\paragraph{Consequence: meshing in linearithmic time.}
Phase~A fills all buckets in $O(\log w)$ steps w.h.p.\ with explicit
constants. For fixed $\delta$, Lemma~\ref{lem:phaseA-all-whp} gives
\[
j_\delta \;\le\; 11.489\,\ln w + O(1) \;=\; 7.96\,\log_2 w + O(1).
\]
Since each iteration maintains at most $O(w)$ candidates, the mesh is
constructed in expected $O(w \log w)$ time.

\paragraph{Comparison to prior work.}

In our work, we consider a parameter $w$ controlling the fineness of the mesh. In the worst case, assuming all buckets are filled, any $x \in [-B/2, B/2]$ would be at most $\Delta=\frac{B}{w}$ away from an anchor.

Da~Cunha et al.~\cite{DaCunhaRSSPMesh} show that given $n$ independent uniform variables $X_1..X_n$, $X_i \in [-1, 1]$, and a constant parameter $\epsilon \in (0, 1/3)$, with probability of at least $1-\epsilon$, for all $z \in [-1, 1]$, there exists a subset $S$ of $X_1..X_n$ such that $|\sum_{i\in S}X_i-z| \le \epsilon$ if $n=O(\log(1/\epsilon))$. This can be restated equivalently as:

Given $n$ independent uniform variables $X_1..X_n$, $X_i \in [-B, B]$, and a constant parameter $w\in(3, \infty)$, with probability of at least $1-1/w$, for all $z \in [-B, B]$, there exists a subset $S$ of $X_1..X_n$ such that $|\sum_{i\in S}X_i-z| \le B/w$ if $n=O(\log w)$.

If one were to naively
materialize the full binary expansion tree to depth $O(C\log w)$, the
number of generated partial sums would be
\[
2^{O(C\log w)} \;=\; w^{O(C)}.
\]
In contrast, our Phase~A construction enforces width $w$ throughout via
bucketing and trimming at every level, and we prove that all $w$ buckets
are filled with probability at least $1-\delta$ within
\[
j_\delta \;\le\; 7.96\,\log_2 w \;+\; 5.19\,\log_2(1/\delta) \;+\; 1
\]
iterations (Lemma~\ref{lem:phaseA-all-whp}). The ability to construct a mesh on an interval $[-B/2, B/2]$ given starting element zero offers a construction to extend to $[-B, B]$, although this is usually not necessary: Build a mesh on $[-B/2, B/2]$, find the minimum anchor in the mesh $a_{min}\leq-B/2-B/w$, then build a mesh on $[a_{min}-B/2, a_{min}+B/2]$. Find the minimum on this mesh $a'_{min}\leq B-2B/w$, creating another mesh $[a'_{min}-B/2, a'_{min}+B/2]$. For $w \geq 4$ the union of these meshes will fully cover $[-B, 0]$. Repeat symmetrically for positive side, for a pessimistic additional constant factor of 5.

Consequently, the mesh is
constructed in $O(w\log w)$ time while maintaining $O(w)$ state.

For any fixed failure probability $\delta$ (e.g.\ $\delta=0.01$), the
leading constant on $\log_2 w$ is explicit and below $10$
($7.96\,\log_2 w+O(1)$), whereas the result of Da Cunha et al. does not provide
an explicit, bounded constant for a width-$w$ trimmed construction.

This result strengthens proofs that rely on the results of Da Cunha et al, such as the proof of the Strong Lottery Ticket Hypothesis \cite{FrankleLotteryTicket}.

\begin{lemma}[Uniformity within a bucket]
\label{lem:bucket-uniform}
Fix a time step $t$ and a bucket $\mathcal I_i\subseteq[-B/2,B/2]$ with $0\notin \mathcal I_i$.
Let $E_{t,i}$ be the multiset of candidate anchors that fall in $\mathcal I_i$ at time $t$,
i.e.,
\[
E_{t,i} \;:=\; \{\,p+s_t : p\in A^{(t-1)},\; p+s_t\in \mathcal I_i\,\},
\]
where $s_t\sim U([-B,B])$ is independent of $A^{(t-1)}$.
If the algorithm selects $a$ by choosing an element uniformly at random from $E_{t,i}$ (conditioned on $E_{t,i}\neq\emptyset$),
then
\[
a \;\big|\; (a\in \mathcal I_i,\;E_{t,i}\neq\emptyset) \;\sim\; U(\mathcal I_i).
\]
\end{lemma}

\begin{proof}
Condition on $A^{(t-1)}$ and on the event that a specific parent $p\in A^{(t-1)}$
produces a candidate in $\mathcal I_i$, i.e.\ on the event
\[
F_p := \{\,p+s_t\in \mathcal I_i\,\} \equiv \{\,s_t\in \mathcal I_i-p\,\}.
\]
Because $s_t\sim U([-B,B])$ and $\mathcal I_i\subseteq[-B/2,B/2]$ while $p\in[-B/2,B/2]$
(by construction of Phase~A), we have $\mathcal I_i-p\subseteq[-B,B]$.
Therefore, conditioning on $F_p$ yields
\[
s_t \mid F_p \;\sim\; U(\mathcal I_i-p).
\]
By translation, this implies
\[
p+s_t \mid F_p \;\sim\; U(\mathcal I_i).
\]

Now condition only on the event that the bucket is nonempty, $E_{t,i}\neq\emptyset$.
Each element of $E_{t,i}$ arises from some parent $p$ and, conditional on its existence,
has distribution $U(\mathcal I_i)$ as shown above.
The algorithm selects uniformly among the (random number of) elements in $E_{t,i}$, so the selected $a$
is a mixture of $U(\mathcal I_i)$ distributions with mixing weights that sum to $1$.
A mixture of identical distributions is the same distribution; hence
$a\sim U(\mathcal I_i)$.
\end{proof}

\subsection*{Phase B}

Let $\mathcal{W}$ denote the Phase~B beam. Assume Phase~A has filled a constant fraction of buckets, and for simplicity that all buckets are filled. This happens in $O(\log w)$ steps, as proven in Phase A analysis.
To simplify analysis, we delete every even-indexed bucket; this guarantees anchors are separated by $\Theta(B/w)$ and that the mesh (maximum spacing) is also $\Theta(B/w)$. This is done because each bucket has length $O(B/w)$ so having no two adjacent buckets guarantees at least that much spacing. Deleting every other bucket can only reduce maximum spacing by a constant factor, and with all buckets filled max spacing was $O(B/w)$.

Without this deletion, anchors could cluster at bucket edges, complicating the spacing guarantees.

Define 
\[
Z \;=\;\{\,T-a_i : a_i\in A\,\}. 
\] 
For any $x\in \mathcal{W}$, its distance to $Z$ is 
\[
d(x) \;=\; \min_{z\in Z}\,|x-z|.
\]
If $d(x)=d$, then choosing a shift $s\sim U([-B,B])$ that lands in $[z-d,\,z+d]$ for some $z\in Z$ reduces the error to at most $d' < d$.  

---
\subsubsection*{Burn-in to the small-gap regime}

Before the expected error decay can take effect, the beam must reach the target range and populate the anchor Voronoi cells. We call this the ``burn-in'' period.

\begin{lemma}[Burn-in Duration]
\label{lem:burnin-summary}
Assume Phase~A has filled all buckets, yielding an anchor set $Z$ spaced by $\Theta(B/w)$. With high probability, every anchor Voronoi cell will be occupied by at least one beam element after
\[ t_{\text{burn}} = O\!\left(\frac{\min Z}{B}\right) + O(\log w). \]
steps.
\end{lemma}

\begin{proof}[Proof sketch]
The burn-in occurs in two distinct phases:
\begin{enumerate}
    \item \textbf{Reaching the anchor range:} Before entering the target range $Q = [\min Z - c\frac{B}{w}, \max Z + c\frac{B}{w}]$, the beam scoring strictly preserves the largest values. The beam evolves as a random walk with positive drift. By a standard Central Limit Theorem argument, the beam enters $Q$ in $O(\min Z / B)$ steps.
    \item \textbf{Filling Voronoi cells:} Once inside $Q$, the multi-target scoring ensures that when a Voronoi cell is occupied, it cannot become vacant again. The filling of the remaining unoccupied cells follows the exact same discrete logistic drift dynamics as Phase A (see Section~\ref{sec:phaseA-logw-coverage}). Consequently, all $m = \Theta(w)$ cells are filled in an additional $O(\log w)$ steps.
\end{enumerate}
The rigorous proofs establishing error monotonicity, filling monotonicity, and the formal logistic drift lower bound are deferred to Appendix~\ref{app:burn-in}.
\end{proof}

\subsubsection*{Single element analysis}

Fix a beam element $x$ at distance $d(x)=d$ from $Z$.

\paragraph{Probability of improvement.}
- Each anchor $z\in Z$ contributes an interval $[z-d,\,z+d]$ of length $2d$.  
- Because anchors are $\Omega(B/w)$ apart, for Lebesgue measure calculations, these intervals can be considered disjoint, up to constant factors.
- There are $\Theta(w)$ anchors within range $B$ (since spacing is $\Theta(B/w)$, and the whole domain has length $2B$).  
- Thus the total “improvement region” has length $\Theta(wd)$.  
- Since $s\sim U([-B,B])$, the probability of improvement is
\[
\Pr(\text{improvement}\mid d) \;=\; \Theta\!\left(\min\{1, \frac{wd}{B}\}\right).
\]
\begin{lemma} \label{lem:uniform-improvement}
Conditional on an improvement occurring, the new minimum gap $D'$ is stochastically dominated by a uniform distribution on $[0,D]$. That is, conditional on improvement, the expected new gap satisfies $\mathbb{E}[D'] \le \frac{D}{2}$.
\end{lemma} 
\begin{proof}
For each pair $(i,k)$ of beam element $x_i$ and anchor $Z_k$, an improvement occurs if the shift $s$ falls within the interval $[Z_k - x_i - D, Z_k - x_i + D]$, which has length $2D$. On any such single interval, the mapping
\[s \;\longmapsto\; |(x_i + s) - Z_k|.\]
is piecewise linear, and its image is exactly $[0,D]$. Because $s$ is drawn uniformly from $[-B, B]$, if $s$ falls into exactly one such improvement interval, the resulting new gap is exactly uniform on $[0,D]$.

In the case where multiple improvement intervals overlap, $s$ may fall into an intersection of two or more intervals. When this happens, the new global minimum gap $D'$ is the minimum over all pairs $(i,k)$ for which $s$ is an improving shift. Since taking the minimum of multiple variables can only decrease the final value, the resulting distribution of $D'$ is stochastically smaller than $U(0, D)$. Therefore, overlapping intervals strictly improve the error decay, and we can safely upper-bound the expected new gap by the expectation of a $U(0, D)$ variable, yielding $\mathbb{E}[D'] \le \frac{D}{2}$.
\end{proof}
\paragraph{Expected drift.}

Let $D_t$ be the global minimum gap at step $t$, where $D_0$ is the minimum gap right after Burn-in, and hence $D_0 \leq B/w$.

Because we are bounding the expected remaining error from above, the gap shrinks by at least:
\[
\mathbb{E}[\,D_{t+1}-D_t \mid D_t=D\,] 
\;\leq\; -\,c\,\frac{w}{B}\,D^2
\]
for some absolute constant $c>0$.

---
\subsubsection*{Discrete Decay Bound}
By definition, after burn-in, the initial gap satisfies $D_0 \leq B/w$. Setting $k = c\frac{w}{B}$, we have $k D_0 \le c$. By absorbing constants into the base step if necessary, we can ensure the strict condition $k D_0 < 1$ is met to apply the following inductive bound.

\begin{lemma}[Discrete Decay Error Bound]\label{lem:discrete-upper-single}
Assume the expected gap is bounded by the discrete recursion envelope
\[D_{t+1} = D_t - k D_t^2,\]
where $k = c\frac{w}{B}$ and the initial state satisfies $k D_0 < 1$. 
Then $D_t = O\!\left(\frac{B}{wt}\right)$.
\end{lemma}
\begin{proof}
Rearranging this recursion yields:
\[\frac{1}{D_{t+1}} = \frac{1}{D_t(1 - k D_t)} \geq \frac{1}{D_t} + k,\] 
where the inequality follows from the geometric series expansion of $(1 - k D_t)^{-1}$. This expansion is strictly valid because $k D_t \le k D_0 < 1$. By induction over $t$ steps, this yields:
\[ \frac{1}{D_{t+1}} \geq \frac{1}{D_0} + k(t+1) \implies D_{t+1} \leq \frac{D_0}{1 + k D_0(t+1)}. \]
Substituting $k = c\frac{w}{B}$ back into the bound, we obtain:
\[ D_{t+1} \leq \frac{D_0}{1 + c\frac{w D_0}{B}(t+1)} = \frac{1}{\frac{1}{D_0} + c\frac{w}{B}(t+1)}. \]
Using the initial bound $D_0 \leq B/w$, we know $\frac{1}{D_0} \geq \frac{w}{B}$. Substituting this into the denominator strictly upper-bounds the fraction, yielding:
\[D_{t+1} \leq \frac{1}{\frac{w}{B} + c\frac{w}{B}(t+1)} = \frac{1}{\frac{w}{B}\big(1 + c(t+1)\big)}. \] 
Since $c > 0$ is a constant, this simplifies directly to the asymptotic bound:
\[ D_t = O\!\left(\frac{B}{wt}\right), \]
yielding the stated decay rate.
\end{proof}

\begin{lemma}{Deterministic Recursion Overestimates Error Compared to Stochastic Recursion}\label{lem:stochastic-leq-discrete}
The deterministic recursion
\[D_{t+1} = D_t - k D_t^2\] 
can only overestimate the error compared to the stochastic version
\[\mathbb{E}[D_{t+1} \mid D_t] \le D_t - k D_t^2.\]
\end{lemma}
\begin{proof}
We proceed by analyzing the unconditional expectation of the stochastic drift. Taking the expectation of both sides with respect to the filtration up to time $t$ yields:
\[ \mathbb{E}[D_{t+1}] \le \mathbb{E}[D_t] - k \mathbb{E}[D_t^2]. \]
Because the function $f(x) = x^2$ is strictly convex, we apply Jensen's inequality, which guarantees that the expectation of the square is bounded below by the square of the expectation: $\mathbb{E}[D_t^2] \ge (\mathbb{E}[D_t])^2$. 

Since $k > 0$, substituting this lower bound into the subtracted term strictly upper-bounds the right-hand side:
\[ \mathbb{E}[D_{t+1}] \le \mathbb{E}[D_t] - k (\mathbb{E}[D_t])^2. \]
Let $y_t = \mathbb{E}[D_t]$ represent the expected error at step $t$. The sequence of expectations satisfies the recurrence inequality:
\[ y_{t+1} \le y_t - k y_t^2. \]
Now, consider the deterministic sequence $z_t$ defined by the exact recurrence $z_{t+1} = z_t - k z_t^2$, initialized at $z_0 = y_0 = \mathbb{E}[D_0]$. 

Because $D_t \ge 0$ almost surely, $y_t \ge 0$. Provided the initial state satisfies the stability condition $2k z_t < 1$ (which guarantees the mapping $x \mapsto x - kx^2$ is strictly monotonically increasing on the relevant domain), we can proceed by induction. Assume $y_t \le z_t$; then:
\[ y_{t+1} \le y_t - k y_t^2 \le z_t - k z_t^2 = z_{t+1}. \]
Therefore, $\mathbb{E}[D_t] \le z_t$ for all $t \ge 0$. Because the stochastic variance term ($\text{Var}(D_t) = \mathbb{E}[D_t^2] - (\mathbb{E}[D_t])^2$) is strictly non-negative, the stochastic sequence will, in expectation, decay faster than its deterministic counterpart. Thus, the deterministic ODE or discrete envelope serves as a rigorous upper bound for the expected stochastic error.
\end{proof}

\subsubsection*{Multiple Elements}

Let the minimum distance of the $w$ beam elements to any anchor be $D$.

At step $j$, the gap $D$ improves iff
\[x_i+s_j \;\in\; \bigcup_{k=1}^w [\,Z_k-D,\;Z_k+D\,] \;=:\; A \quad\text{for some } i\in[w],\]
or equivalently
\[s_j \;\in\; \bigcup_{i=1}^w (A - x_i).\]

Hence, the probability of improvement relies on the total Lebesgue measure of this union:
\[\Pr(\text{improvement from } D) \;=\; \frac{\operatorname{Leb}\!\left(\bigcup_{i=1}^w (A-x_i)\;\cap\;[-B,B]\right)}{2B}.\]

Because Phase A and Phase B are generated via tree expansions, the precise joint distribution of the anchors $Z$ and the beam elements $x$ contains path dependencies. However, to evaluate the expected volume of the improvement region, we only require that the $w^2$ pairwise differences $Z_k - x_i$ do not perfectly overlap. Under \textbf{Assumption 5}, the local offsets of these elements are sufficiently decorrelated such that we can treat them as pairwise independent for the purpose of the union bound, preventing a state-space collapse.

\begin{lemma}[Anchor Convolution Union: Expectation Lower Bound]
\label{lem:anchor-conv}
Partition $[-B/2,B/2]$ into $w$ buckets of length $\Delta:=\tfrac{B}{w}$. Under Assumption 5, let the anchors and beam elements be modeled as:
\[Z^{(1)}_i = C_i+U_i,\qquad Z^{(2)}_j = C_j+U_j+D_j,\]
where $C_i=-B/2+\Big(i-\tfrac12\Big)\Delta$, the variables $U_i, U_j \sim \mathrm{Unif}[-\Delta/2,\Delta/2]$ represent the pairwise-decorrelated microscopic offsets, and $D_j$ are arbitrary random variables supported on $[-\Delta/2,\Delta/2]$ representing the current gap. 

Define the cross-differences $D_{ij}=Z^{(1)}_i-Z^{(2)}_j$ and the union of improvement intervals:
\[\mathcal U_d \;=\; \bigcup_{i,j=1}^w [\,D_{ij}-d,\,D_{ij}+d\,].\]
There exist absolute constants $c,c'>0$ such that if $d \;\le\; c\,\frac{\Delta}{w}$, then the expected measure avoids degenerate collapse and scales quadratically with $w$:
\[\mathbb{E}\big[\,\mathrm{Leb}(\mathcal U_d)\,\big]\;\ge\; c'\,w^2\,d.\]
\end{lemma}

\begin{proof}
Write the cross-difference as:
\[D_{ij} = (C_i-C_j) + (U_i-U_j) - D_j.\]
Group pairs $(i,j)$ by the macroscopic diagonal offset $s=i-j\in\{-(w-1),\dots,w-1\}$. For a fixed $s$, the deterministic macroscopic distance $C_i-C_{i-s}=s\Delta$ is constant, leaving the microscopic jitter governed strictly by $(U_i-U_{i-s}) - D_{i-s}$. Because Assumption 5 grants pairwise independence to the offsets $U$, the distribution of $(U_i-U_{i-s})$ forms a non-degenerate convolution (a triangular distribution of width $2\Delta$). 

The number of pairs for a given $s$ is $M_s=w-|s|$. For any $|s| \le w/2$, we have $M_s = \Theta(w)$. Since the centers $D_{ij}$ are independently and continuously smeared across intervals of width $2\Delta$, the expected intersection length of any two distinct target intervals $E_a$ and $E_b$ on the same diagonal is strictly bounded by their convolution: $\mathbb{E}[\mathrm{Leb}(E_a \cap E_b)] = O(d^2/\Delta)$. Applying the second-moment Bonferroni inequality, the expected measure of the union is bounded below by the sum of individual expected lengths minus the sum of pairwise expected intersections:
\[\mathbb{E}\left[\mathrm{Leb}\left(\bigcup_{k=1}^{M_s} E_k\right)\right] \ge \sum_{k=1}^{M_s} 2d - \sum_{a < b} O\left(\frac{d^2}{\Delta}\right) = 2M_s d - O\left(M_s^2 \frac{d^2}{\Delta}\right).\]
Given the condition $d \le c \frac{\Delta}{w}$ and $M_s \le w$, the subtracted intersection term simplifies to $O\left(M_s d \cdot \frac{wd}{\Delta}\right) \le O(M_s d \cdot c)$. By choosing the constant $c$ sufficiently small, the dominant linear sum is strictly preserved, yielding $\Omega(M_s d)$ per diagonal. Summing across the $\Theta(w)$ primary diagonals yields the global lower bound $\mathbb{E}[\mathrm{Leb}(\mathcal U_d)] = \Omega(w^2 d)$.
\end{proof}

\paragraph{Connection to the improvement probability.}
Let $D$ be the current minimum beam--anchor distance. Improvement occurs iff
\[
s \in \bigcup_{i=1}^w (A(D)-x_i)\cap[-B,B],
\qquad
A(D)=\bigcup_{k=1}^w [Z_k-D,Z_k+D].
\]

First, after burn-in, every $V_j$ is occupied by at least one element, implying that every anchor has an associated beam element attached to it.

Second, by Lemma~\ref{lem:bucket-uniform}, each anchor satisfies 
$Z_k=C_k+U_k$ with $U_k$ uniform within its bucket. 

Thirdly, by Lemma \ref{lem:uniform-improvement}, when a beam element with original distance $d$ attaches to an anchor, its new distance is rolled as $U(0, d)$, independently from the new anchor. By the single-element analysis, $d$ will shrink below $\frac{\Delta}{2}$ in $O(1)$ steps. This can be pessimistically replaced by $\mathcal{D}(0, \Delta/2)$, where $\mathcal{D}$ is an arbitrary distribution. Using offset rather than absolute distance yields $\mathcal{D}(-\Delta/2, \Delta/2)$.

Hence we may pessimistically beam element $x_i$ as follows:
\[
\widetilde x_i = Z_{j(i)}+\mathcal{D}(-\Delta/2, \Delta/2)=C_{j(i)} + U_{j(i)}+\mathcal{D}(-\Delta/2, \Delta/2),
\]
where $j(i)$ is the index of the anchor that beam element $x_i$ is attached to.
$U_i\sim\mathrm{Unif}[-\Delta/2,\Delta/2]$, and $\mathcal{D}$ is an arbitrary distribution, independent from $U_{j(i)}$.

Renaming $Z^{(1)}:=Z$ and $Z^{(2)}_i=C_{j} + U_j+\mathcal{D}(-\Delta/2, \Delta/2)$ allows direct substitution into Lemma \ref{lem:anchor-conv}, yielding
\[
\mathbb{E}\!\left[
  \mathrm{Leb}\!\left(
     \bigcup_{i=1}^w (A(D)-x_i)\cap[-B,B] 
  \right)
\right]
\ge c' w^2 D. 
\]

Dividing by $B$ gives
\[
\Pr(\text{improvement from }D)
\ge
\Theta\!\left(\frac{w^2 D}{B}\right),
\qquad
D \ll B/w^2.
\]

If $D > \frac{B}{w^2}$, note that the improvement Lebesgue measure trivially grows monotonically with $D$. Hence, when $D > \frac{B}{w^2}$,
\[
\mathbb{E}\!\left[
  \mathrm{Leb}\!\left(
     \bigcup_{i=1}^w (A(D)-x_i)\cap[-B,B] 
  \right)
\right]
\ge c' w^2 \frac{B}{w^2}=c'B,
\]
which yields
\[
\Pr(\text{improvement from }D)
\ge
\Omega\!\left(1\right),
\qquad
D > B/w^2,
\]
which causes $D$ to shrink geometrically until $D \ll \frac{B}{w^2}$.

\paragraph{Discrete recursion and transition.}
In the saturated subregime (when $D \gtrsim B/w^2$), the probability of improvement is $\Omega(1)$, causing the gap to decay geometrically by at least a constant factor per step. After at most $T_{sat} = O(\log w)$ steps, the process enters the small-gap regime where $D_{T_{sat}} \ll B/w^2$. 

Once in the small-gap regime, the expected gap shrinks by at least $c'\frac{w^2}{B}D^2$ per step. We define our discrete recurrence constant as $k = c'\frac{w^2}{B}$. Because $D_{T_{sat}} \ll B/w^2$, the entry condition $k D_{T_{sat}} < 1$ is strictly satisfied. By substituting this $k$ into Lemma \ref{lem:discrete-upper-single} and applying to the stochastic version with Lemma \ref{lem:stochastic-leq-discrete}, we obtain an upper bound on the error of:
\[
D_t \;\le\; O\!\left(\frac{B}{w^2 (t - T_{sat})}\right) \;=\; O\!\left(\frac{B}{w^2 t}\right).
\]

\begin{theorem}[Phase~B inverse-quadratic expected error decay]
\label{thm:phaseB-decay}
Consider Phase~B of the MITM beam search with beam width $w$ and anchors
constructed by Phase~A. Assume the burn--in conditions hold, i.e.,
after an initial burn--in of
\[
t_0 = O\!\left(\frac{\min Z}{B}\right) + O(\log w)
\]
steps, every anchor Voronoi cell is occupied by at least one beam element.

Let $D_t$ denote the minimum distance of the Phase~B beam to the residual
anchor set $Z$ at time $t \ge t_0$. Then there exists an absolute constant
$c>0$ such that for all $t \ge t_0$,
\[
\mathbb{E}[D_t] \;\le\; \frac{c\,B}{w^2\,(t-t_0+1)}.
\]

Equivalently, up to constant factors,
\[
\mathbb{E}[D_t] \;=\; O\!\left(\frac{B}{w^2\,t}\right).
\]
Moreover, the standard deviation satisfies
\[
\operatorname{sd}(D_t) \;=\; \Theta\!\left(\mathbb{E}[D_t]\right),
\]
with the variance bound proven in Appendix~A.
\end{theorem}

\begin{proof}[Proof sketch]
After burn--in, every anchor Voronoi cell is occupied. By
Lemma~\ref{lem:anchor-conv}, the probability of improvement from distance
$D_t$ is $\Theta(w^2 D_t/B)$ in the small--gap regime, and bounded below
by a constant when $D_t \gtrsim B/w^2$.

Conditional on improvement, Lemma~\ref{lem:uniform-improvement} shows the
new gap is uniformly distributed on $[0,D_t]$, yielding
\[
\mathbb{E}[D_{t+1}-D_t \mid D_t]
\;=\; -\Theta\!\left(\frac{w^2}{B}\right) D_t^2.
\]
Applying the discrete bounds from Lemma~\ref{lem:discrete-upper-single} via the stochastic comparison in Lemma~\ref{lem:stochastic-leq-discrete} implies
\[
D_t \;\le\; O\!\left(\frac{B}{w^2\,t}\right),
\]
and the stated variance scaling follows from Appendix~A.
\end{proof}

\begin{theorem}[End--to--end MITM beam guarantee]
\label{thm:mitm-main}
Let $S=\{s_1,\dots,s_n\}$ be i.i.d.\ samples from a symmetric distribution
supported on $[-B,B]$, and let $T$ be a target generated as the sum of a subset of $S$. Fix beam width $w$.

Run the MITM beam search with:
\begin{itemize}[leftmargin=*]
\item Phase~A using $n_L = \Theta(\log w + \log(1/\delta))$ elements to
construct anchors, and
\item Phase~B on the remaining $n_R = n-n_L$ elements.
\end{itemize}

\paragraph{Burn-in duration and worst-case targets.}
The total number of elements consumed before Phase~B is $n_{pre} = n_L + t_{burn}$, which evaluates to:
\[ n_{pre} = O(\log w) + O(\log(1/\delta)) + \frac{4\min Z}{B} + O\!\left(\sqrt{\frac{\min Z}{B}}\right) = \frac{4\min Z}{B} + o(n). \]
Because $\min Z \le T$, the survival of Phase~B depends entirely on the magnitude of the target $T$.  However, we now prove that $n-n_{pre}=\Theta(n)$ , unless the instance is solvable in subexponential time.

Let $S_{max} = \sum_{s_i > 0} s_i$ denote the maximum possible subset sum. Because $s_i \sim U([-B, B])$, the expected value of $S_{max}$ is $nB/4$. By Hoeffding's inequality (or CLT), $S_{max}$ concentrates tightly, so $S_{max} = \frac{nB}{4} \pm o(nB)$ with overwhelming probability. Since $T$ is a valid subset sum, $T \le S_{max}$. We divide the target space into two regimes based on an arbitrarily small constant $\gamma \in (0, 1/4)$:

\textbf{Case 1: Strictly bounded targets ($T \le (1/4 - \gamma)nB$).} 
In this regime, the elements consumed before Phase~B are strictly bounded by $n_{pre} \le (1 - 4\gamma)n + o(n)$. This leaves $n_R \ge 4\gamma n - o(n) = \Theta(n)$ elements for Phase~B. Because the number of available Phase~B steps remains linear with respect to $n$, the expected error decay strictly preserves the asymptotic bound of $O\!\left(\frac{B}{n_R w^2}\right) = O\!\left(\frac{B}{n w^2}\right)$, absorbing the constant $\frac{1}{4\gamma}$ into the asymptotic notation.

\textbf{Case 2: Extreme maximum targets ($T > (1/4 - \gamma)nB$).} 
If the target is chosen adversarially close to $S_{max}$, the burn-in may consume $n_{pre} = n - o(n)$ elements, starving Phase~B. However, in this regime, the problem undergoes a severe state-space collapse. Define the residual slack as $\epsilon = S_{max} - T$. Because $T$ is extremely large, $\epsilon = o(nB)$. 

To achieve a sum of $T$, the optimal subset must be formed by taking the maximal configuration and making deviations (either excluding a positive element or including a negative element). Every such deviation consumes a portion of the slack $\epsilon$ equal to the magnitude of the element. Consider the set of "large" elements, $L = \{ s_i \in S \mid |s_i| \ge B/2 \}$. The expected size of $L$ is exactly $n/2 = \Theta(n)$. 

Because each deviation in $L$ consumes at least $B/2$ slack, the maximum number of deviations we can make among the large elements is bounded by:
\[ k_{deviations} \le \frac{\epsilon}{B/2} = \frac{o(nB)}{B/2} = o(n). \]
Therefore, out of the $\Theta(n)$ large elements, at most $o(n)$ can deviate from the maximal configuration. This guarantees that the inclusion/exclusion status of $\Theta(n) - o(n) = \Theta(n)$ elements is deterministically forced. The effective number of undecided elements drops to $o(n)$, allowing the residual problem to be solved exactly via brute-force enumeration or standard dynamic programming in subexponential time. 

Thus, across all valid targets, the algorithm either achieves the $O\!\left(\frac{B}{nw^2}\right)$ error bound via the beam search decay, or the instance collapses into a trivially solvable subexponential state.

Hence, with probability at least $1-\delta$,
\[
\mathbb{E}\bigl[\,|S^\star - T|\,\bigr]
\;=\;
O\!\left(\frac{B}{n\,w^2}\right),
\qquad
\operatorname{sd}\bigl(|S^\star-T|\bigr)
=
\Theta\!\left(\mathbb{E}|S^\star-T|\right).
\]

The algorithm runs in $O(nw\log w)$ time and $O(w)$ memory. Exact subset
reconstruction is supported in $O(nw)$ time using $O(w\sqrt n)$ memory via
checkpointing.
\end{theorem}

\section{Experimentation}
  
We use experiments to verify the claims made in the Phase A/B analysis and in Section 5.2.1, along with providing a comparison to the heuristics in the current RSSP and general literature. 

\subsection{Different Input Distributions}

To evaluate robustness beyond the uniform model used in the analysis, we run the proposed method on several i.i.d.\ input distributions with substantially different shapes (e.g., multimodal, heavy-tailed, and approximately Gaussian). Notably, we include Student's $t$-distributions with low degrees of freedom ($\nu \in \{1, 2\}$) to test extreme heavy-tailed behavior where the input variance is infinite or undefined. Across all tested distributions, the empirical error--runtime curves exhibit the same qualitative scaling behavior predicted by theory; differences are primarily in constant factors, typically within an order of magnitude. This supports the claim that the method’s decay rate is largely distribution-insensitive.

\begin{figure}[htbp]
    \centering
    \includegraphics[width=1\linewidth]{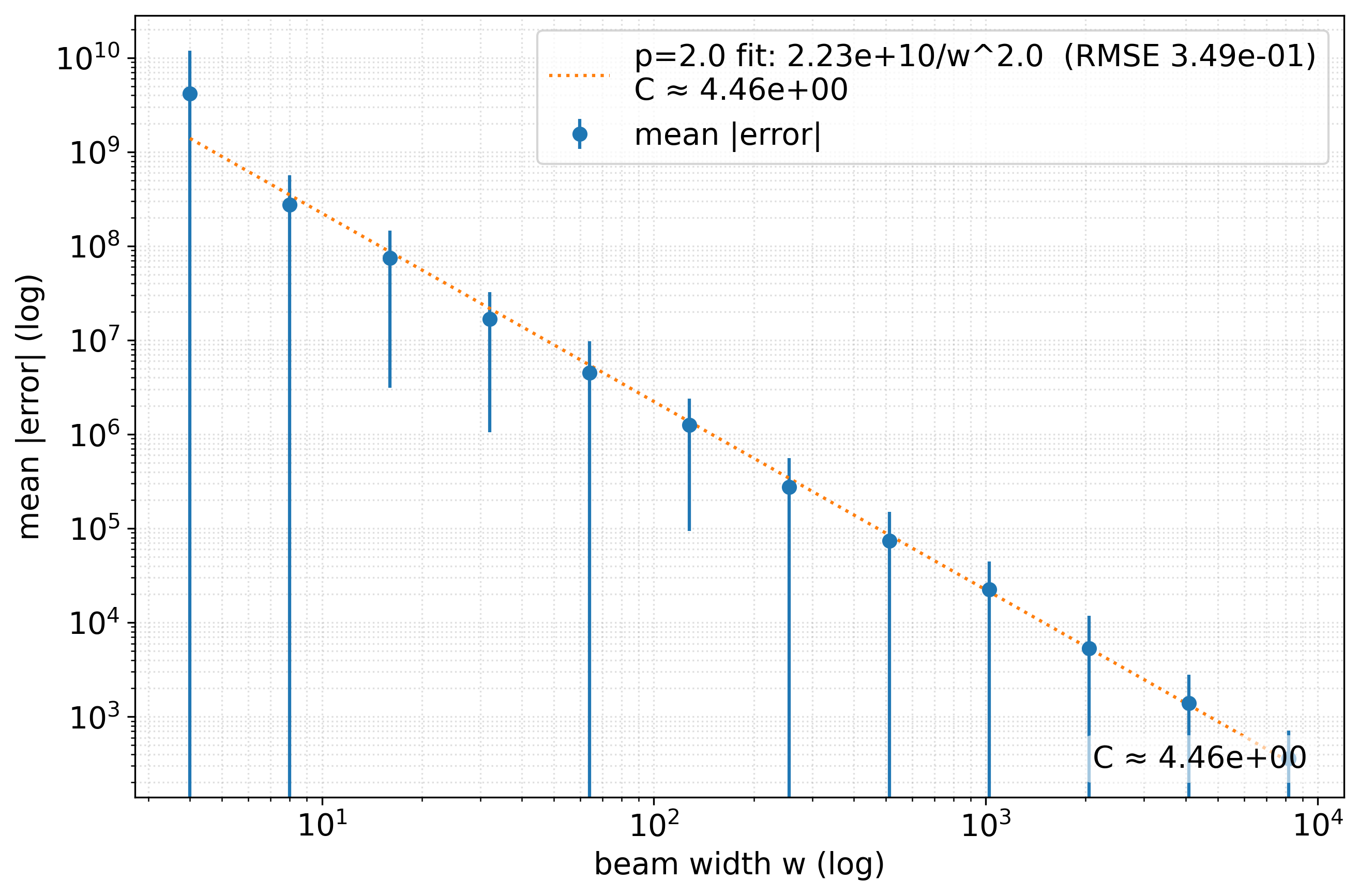}
    \caption{Performance of the proposed method on a symmetric bimodal input distribution. Error scaling remains comparable to the baseline distribution up to constant factors, indicating robustness to multimodality.}
    \label{fig:dist_bimodal_sym}
\end{figure}

\begin{figure}[htbp]
    \centering
    \includegraphics[width=1\linewidth]{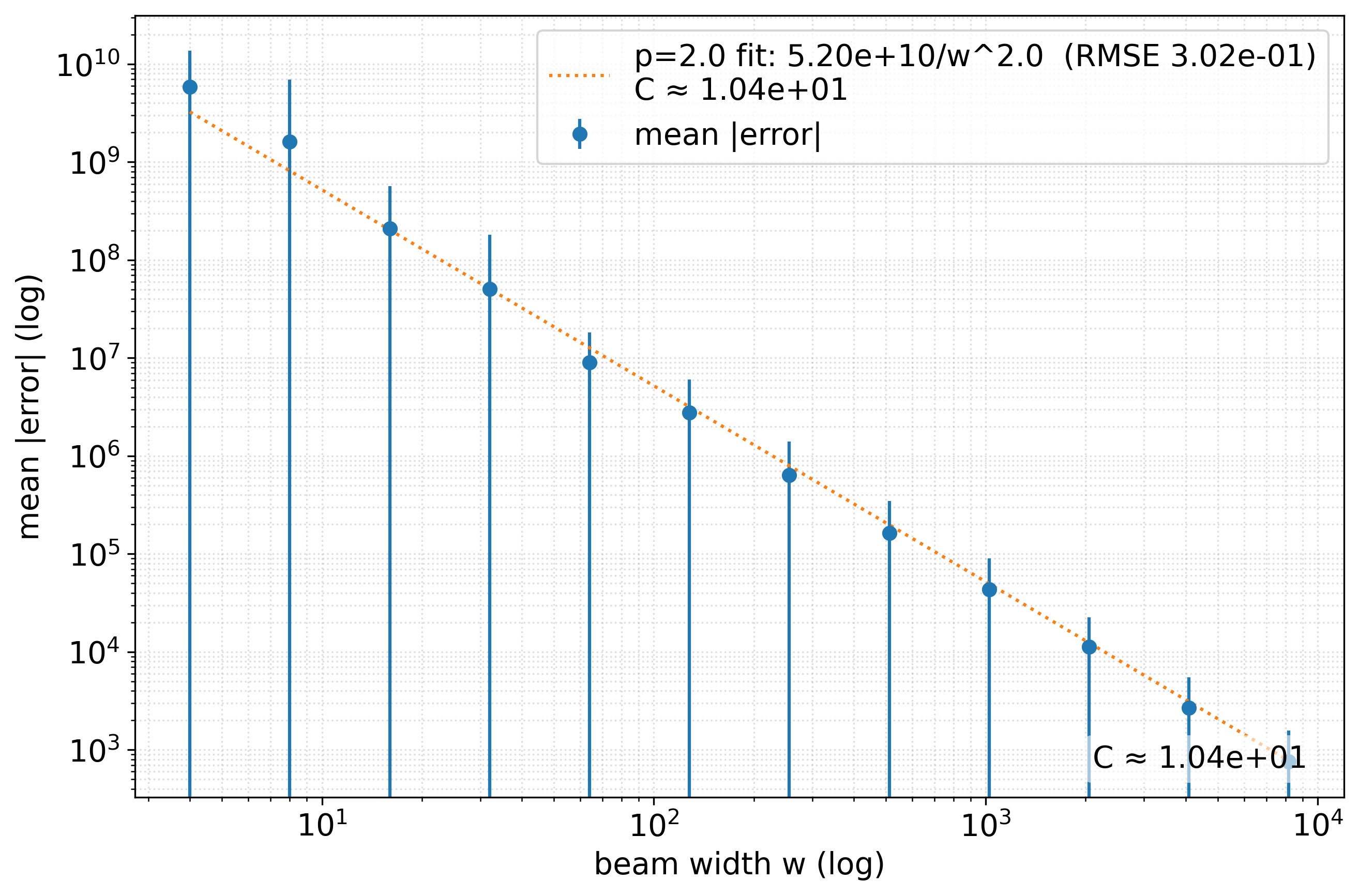}
    \caption{Performance of the proposed method on a symmetric lognormal input distribution. Despite the heavy-tailed nature of the inputs, the observed error scaling closely follows the theoretical predictions.}
    \label{fig:dist_lognormal_sym}
\end{figure}

\begin{figure}[htbp]
    \centering
    \includegraphics[width=1\linewidth]{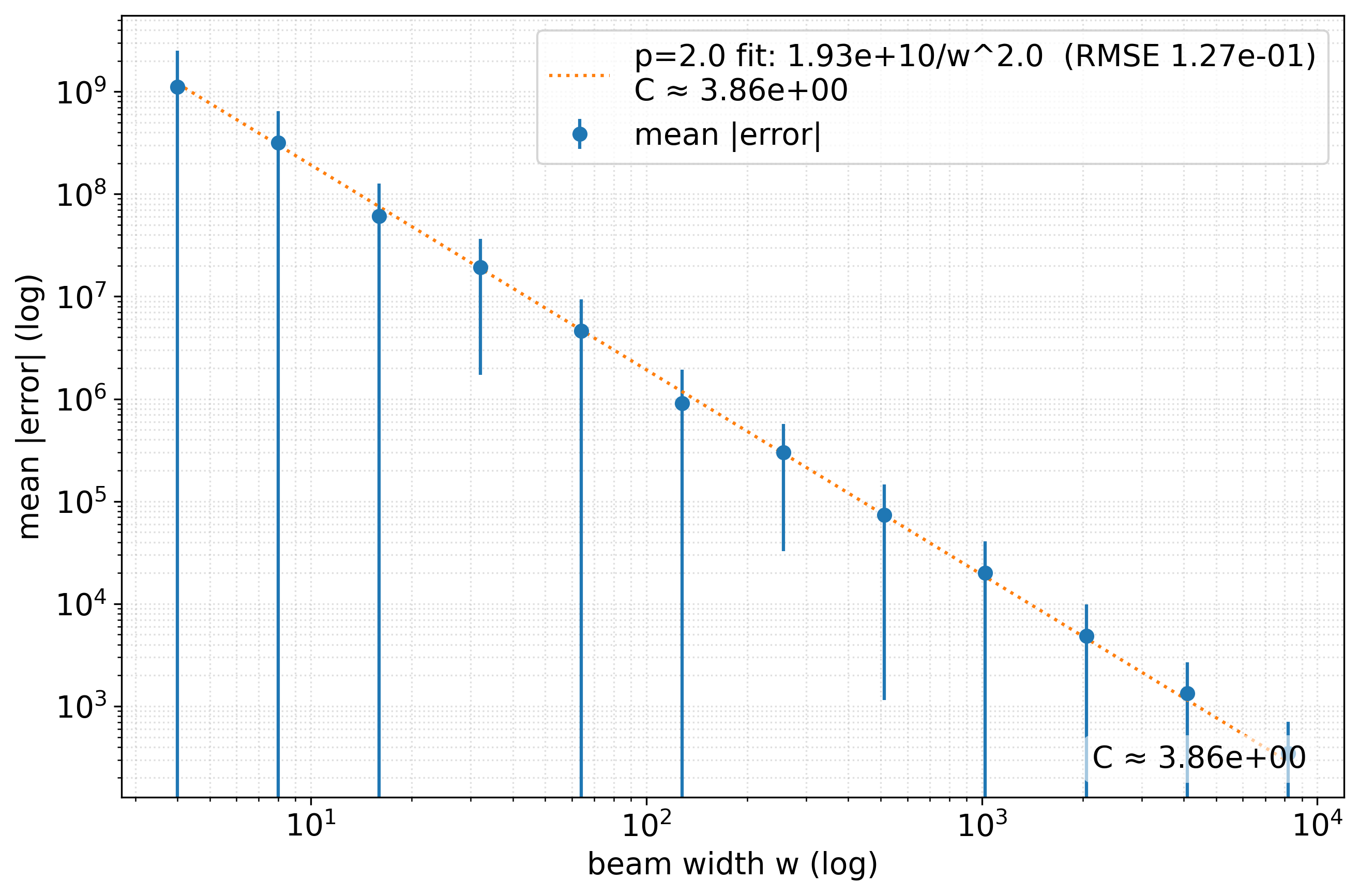}
    \caption{Performance of the proposed method on a symmetric normal input distribution. Results closely match those obtained under other distributions, suggesting weak dependence on the specific input distribution.}
    \label{fig:dist_normal_sym}
\end{figure}

\begin{figure}[htbp]
    \centering
    \includegraphics[width=1\linewidth]{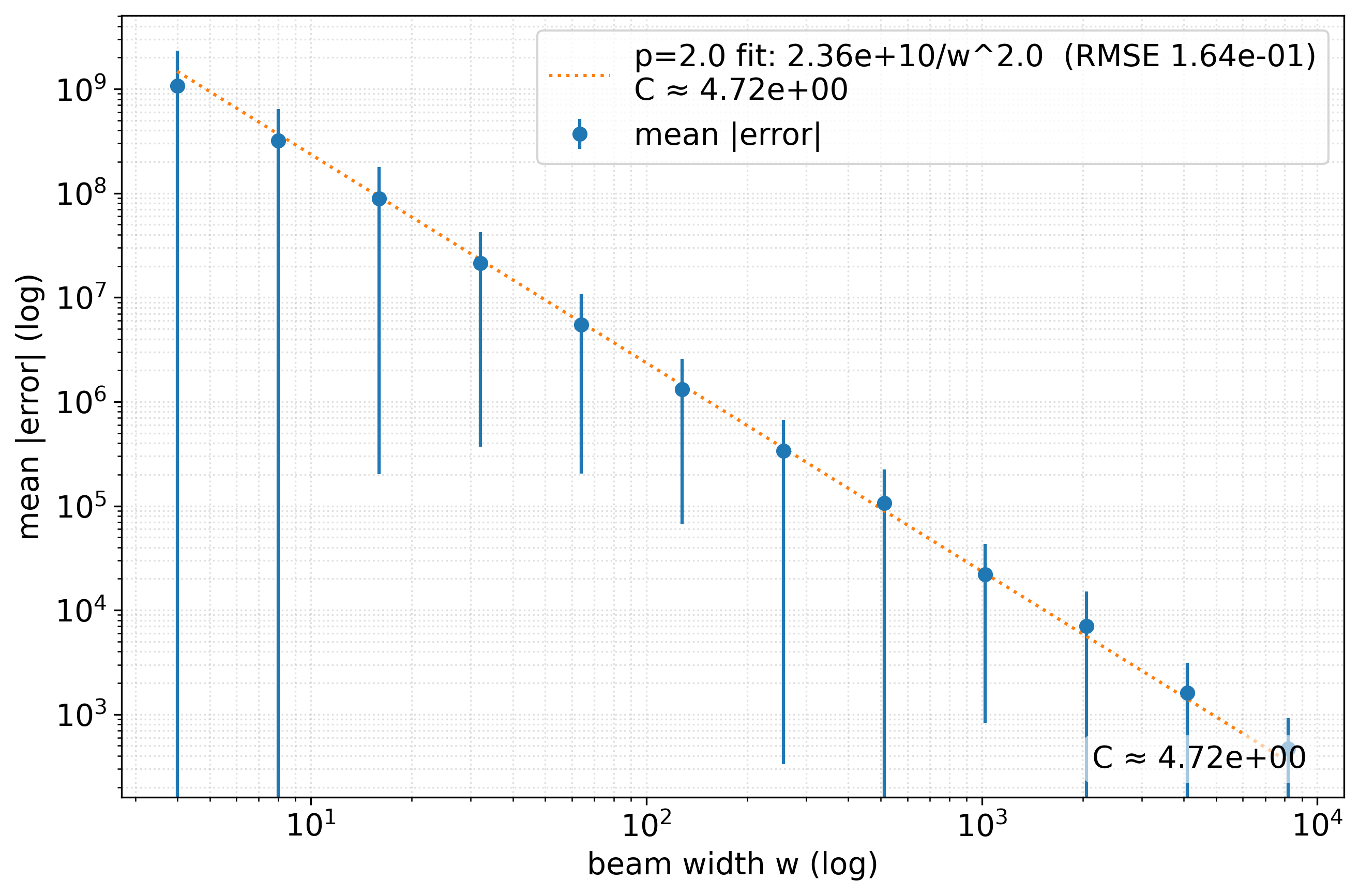}
    \caption{Performance of the proposed method on a symmetric Cauchy input distribution (Student's $t$ with $\nu=1$). Even with an undefined mean and variance, the inverse-quadratic decay rate is preserved.}
    \label{fig:dist_t_df1}
\end{figure}

\begin{figure}[htbp]
    \centering
    \includegraphics[width=1\linewidth]{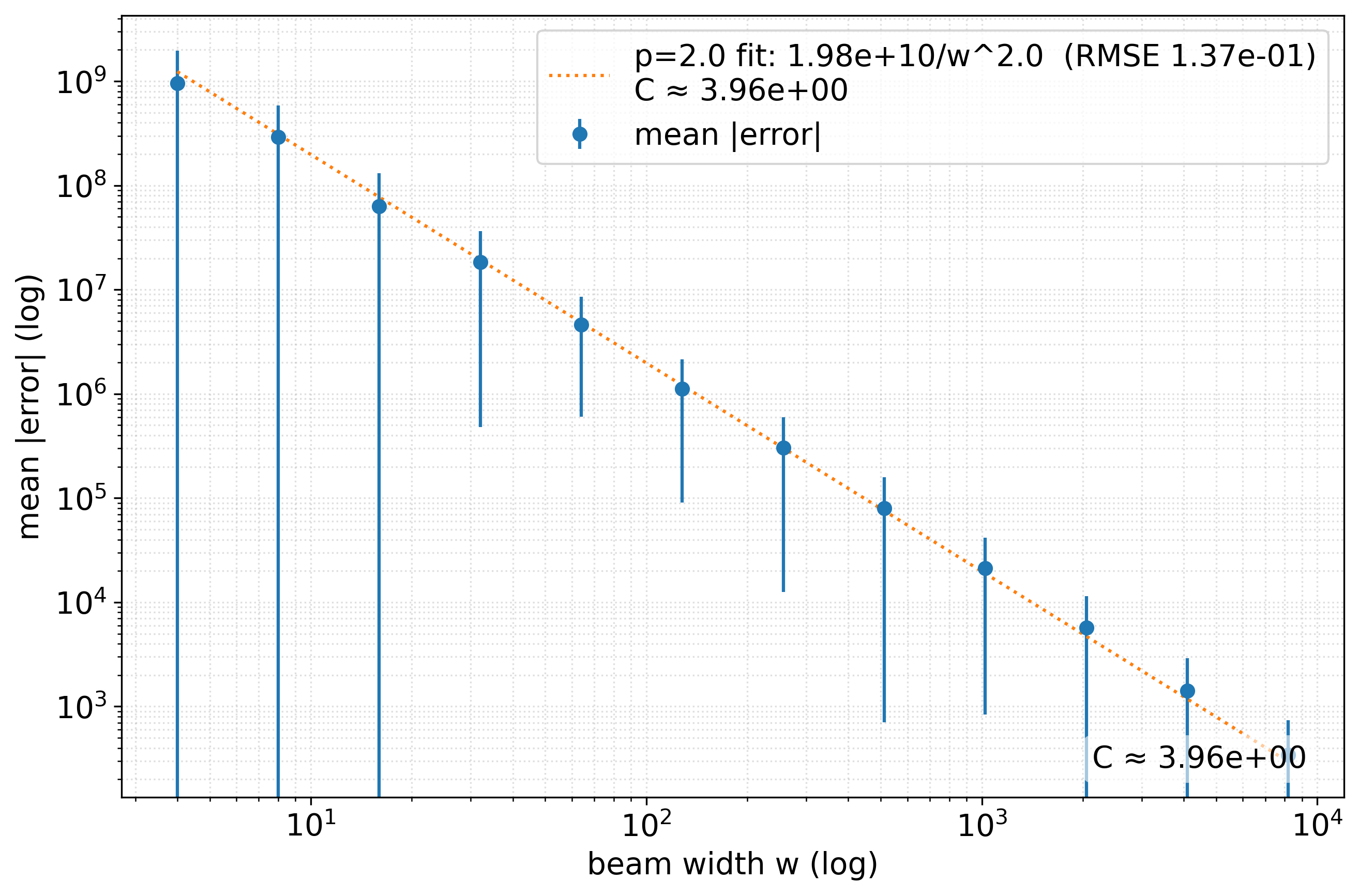}
    \caption{Performance on a symmetric Student's $t$ input distribution with $\nu=2$. The infinite variance of the inputs does not break the $O(w^{-2})$ theoretical error scaling.}
    \label{fig:dist_t_df2}
\end{figure}

\begin{figure}[htbp]
    \centering
    \includegraphics[width=1\linewidth]{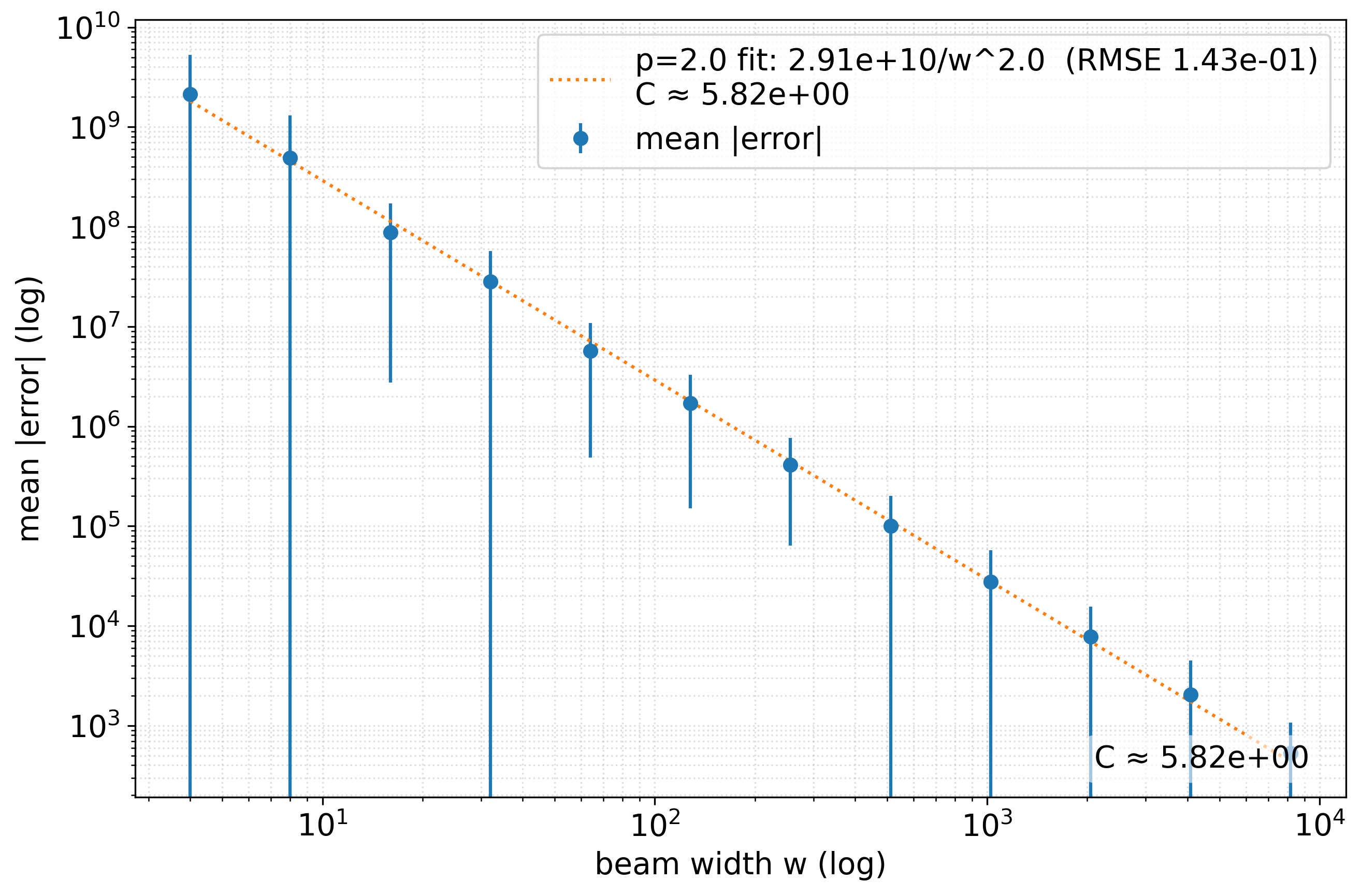}
    \caption{Performance of the proposed method on a uniform input distribution (baseline setting). Observed scaling matches the theoretical prediction up to constants.}
    \label{fig:dist_uniform_baseline}
\end{figure}

\subsection{Different Split Points}

The MITM variant introduces a design choice: how many elements to allocate to Phase~A (anchor construction) versus Phase~B (multi-target beam refinement). The analysis predicts that allocating $O(\log w)$ elements to Phase~A is sufficient to obtain an $O(B/w)$ anchor mesh while preserving enough remaining steps for Phase~B to drive the inverse-quadratic decay. We validate this prediction by comparing the theoretically motivated $O(\log w)$ split against a naive halfway split and an intentionally undersized $O(1)$ split.

\begin{figure}[htbp]
    \centering
    \includegraphics[width=1\linewidth]{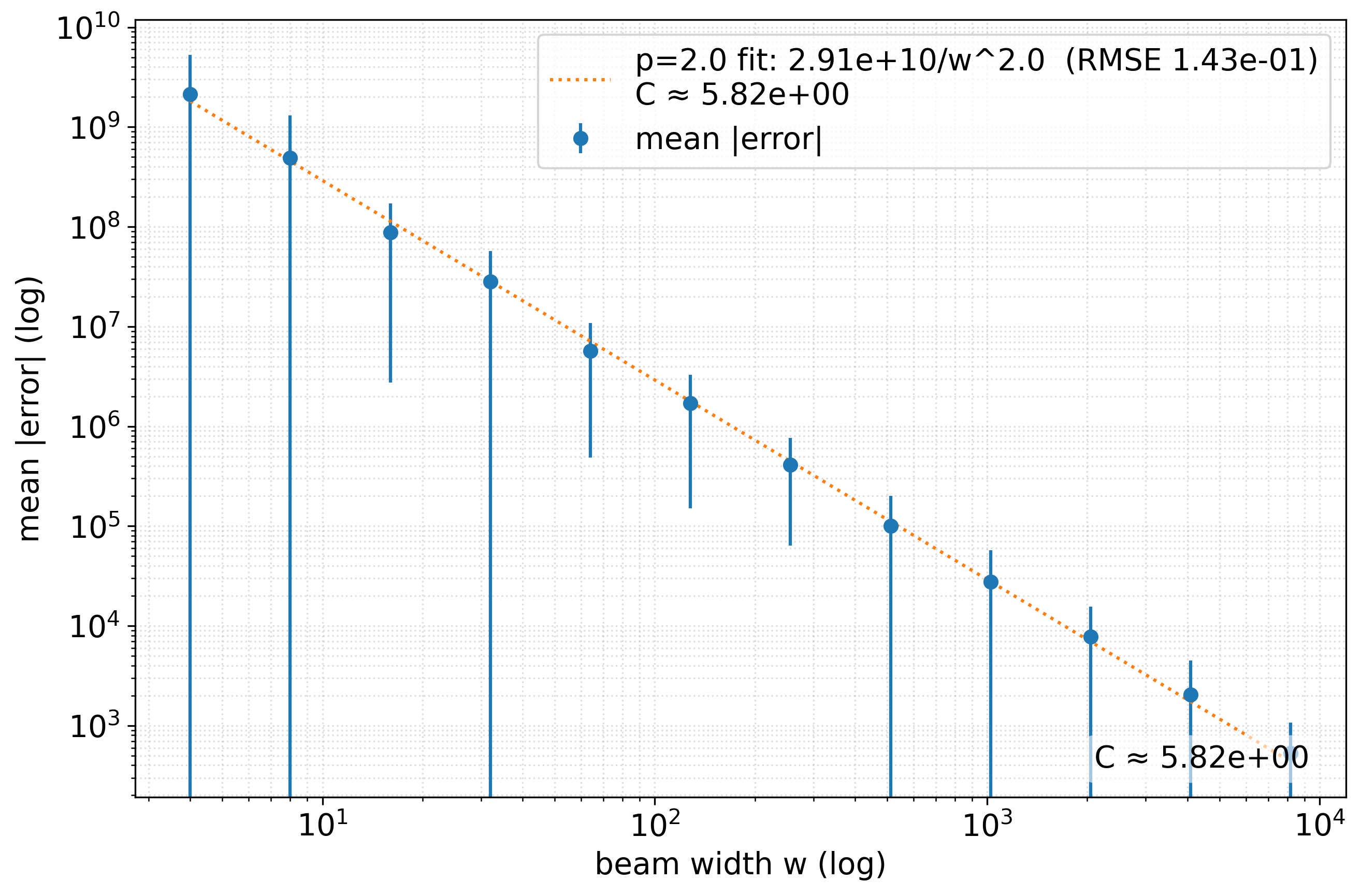}
    \caption{Effect of splitting at the theoretically motivated $O(\log w)$ point. This choice achieves the best empirical trade-off between runtime and approximation error, in line with the analytical guarantees.}
    \label{fig:split_clogw}
\end{figure}

\begin{figure}[htbp]
    \centering
    \includegraphics[width=1\linewidth]{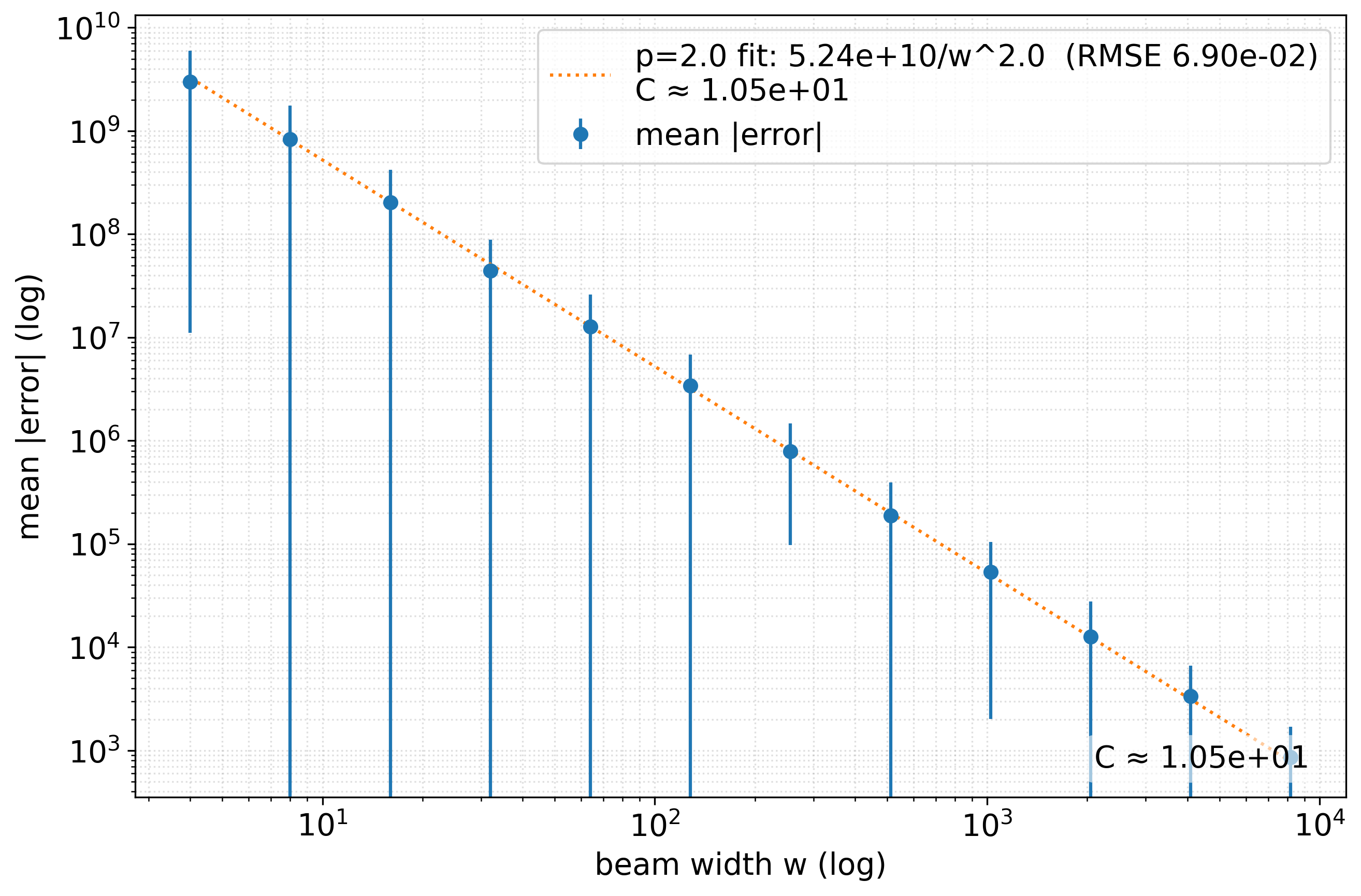}
    \caption{Effect of splitting at the halfway point. This heuristic split is competitive but typically underperforms the $O(\log w)$ split, reflecting the benefit of the theoretically guided partition.}
    \label{fig:split_half}
\end{figure}

\begin{figure}[htbp]
    \centering
    \includegraphics[width=1\linewidth]{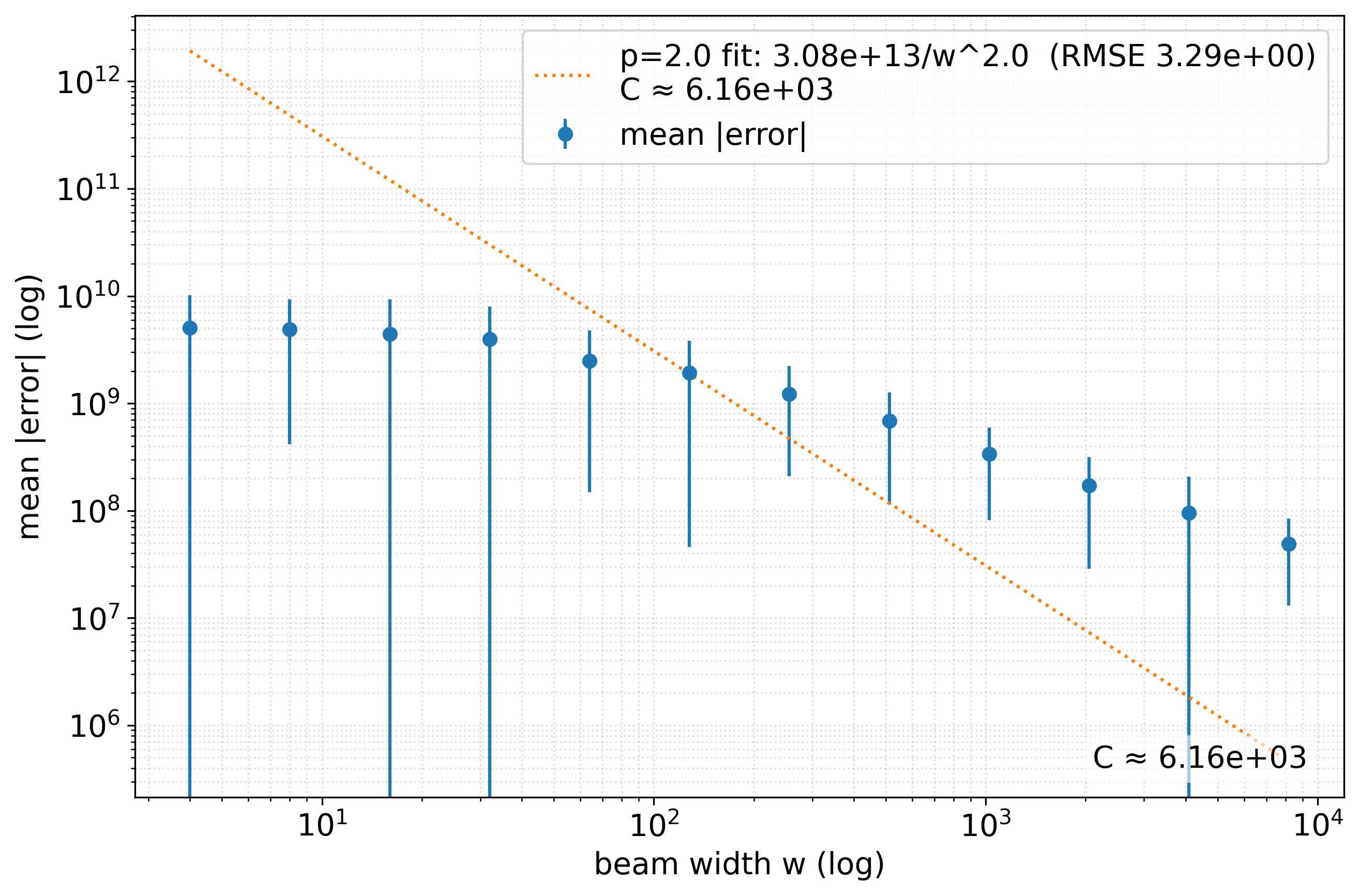}
    \caption{Failure mode when the split is too small (e.g., $O(1)$). With insufficient effort spent in Phase~A, anchors are too sparse and Phase~B cannot enter the multi-target small-gap regime; empirically the decay degrades to roughly inverse-linear scaling.}
    \label{fig:split_too_small_failure}
\end{figure}

\subsection{Dependence on $n$}

Although not essential to the paper, we analyze the effect of the amount of elements on the performance of the algorithm. 
Firstly, we analyze the "medium-$n$" regime, where $n - n_{pre} = \Theta(n)$, but the logarithmic cost of Phase A and burn-in are not negligible. In this case, systematic bias can be seen in the decay, as the burn-in cost becomes less significant. This is shown in Figure~\ref{fig:med-n-dependency}.

In the "large-$n$" regime, where $n \gg n_{pre}$, burn-in and Phase A cost is negligible, yielding a clean $1/n$ dependency, as shown in Figure~\ref{fig:large-n-dependency}.

\begin{figure}[htbp]
    \centering
    \includegraphics[width=1\linewidth]{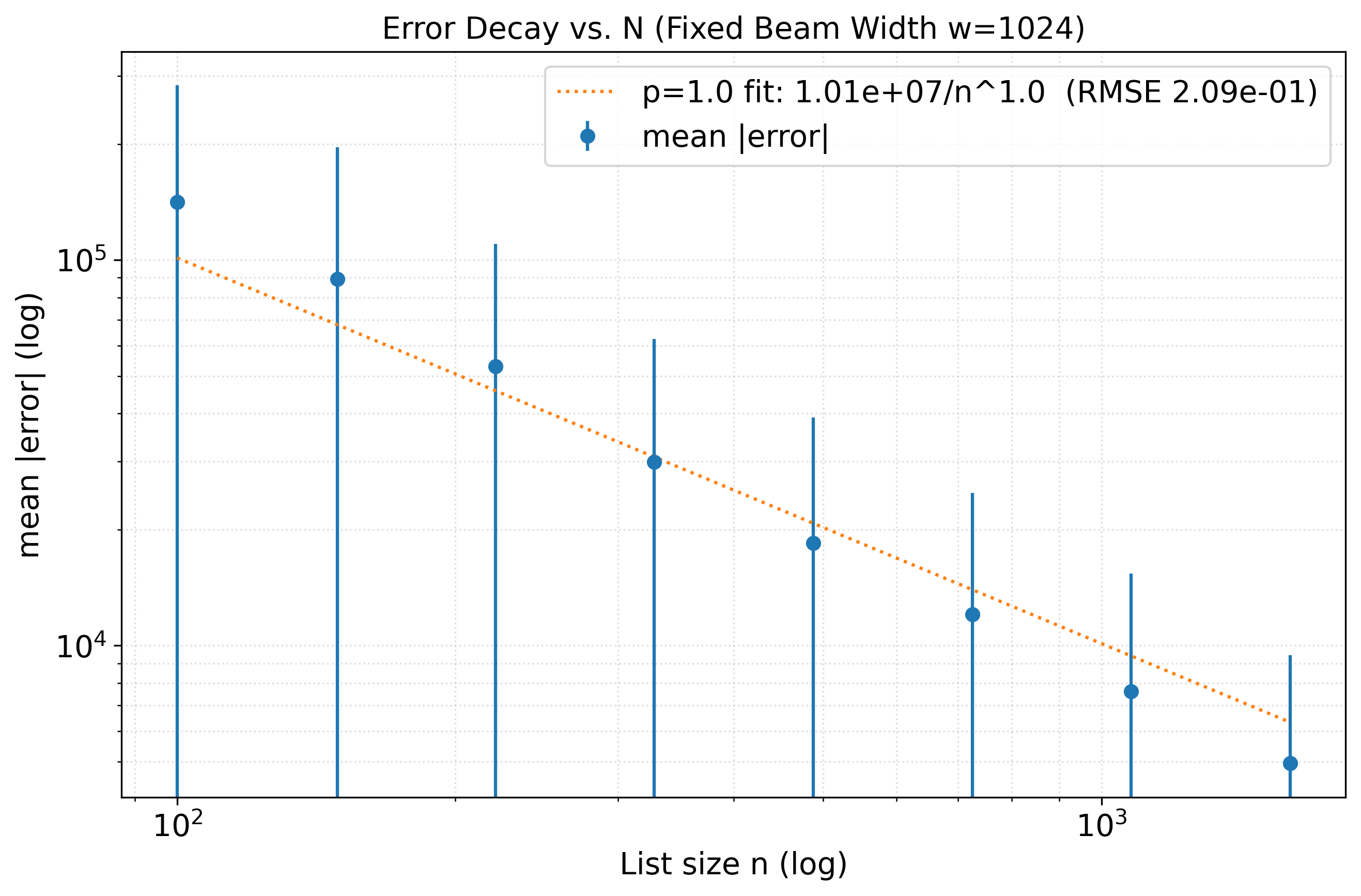}
    \caption{Effect of the number of elements ($n$) on empirical error scaling in the medium-$n$ regime. Systematic bias is visible due to non-negligible Phase A and burn-in costs.}
    \label{fig:med-n-dependency}
\end{figure}

\begin{figure}[htbp]
    \centering
    \includegraphics[width=1\linewidth]{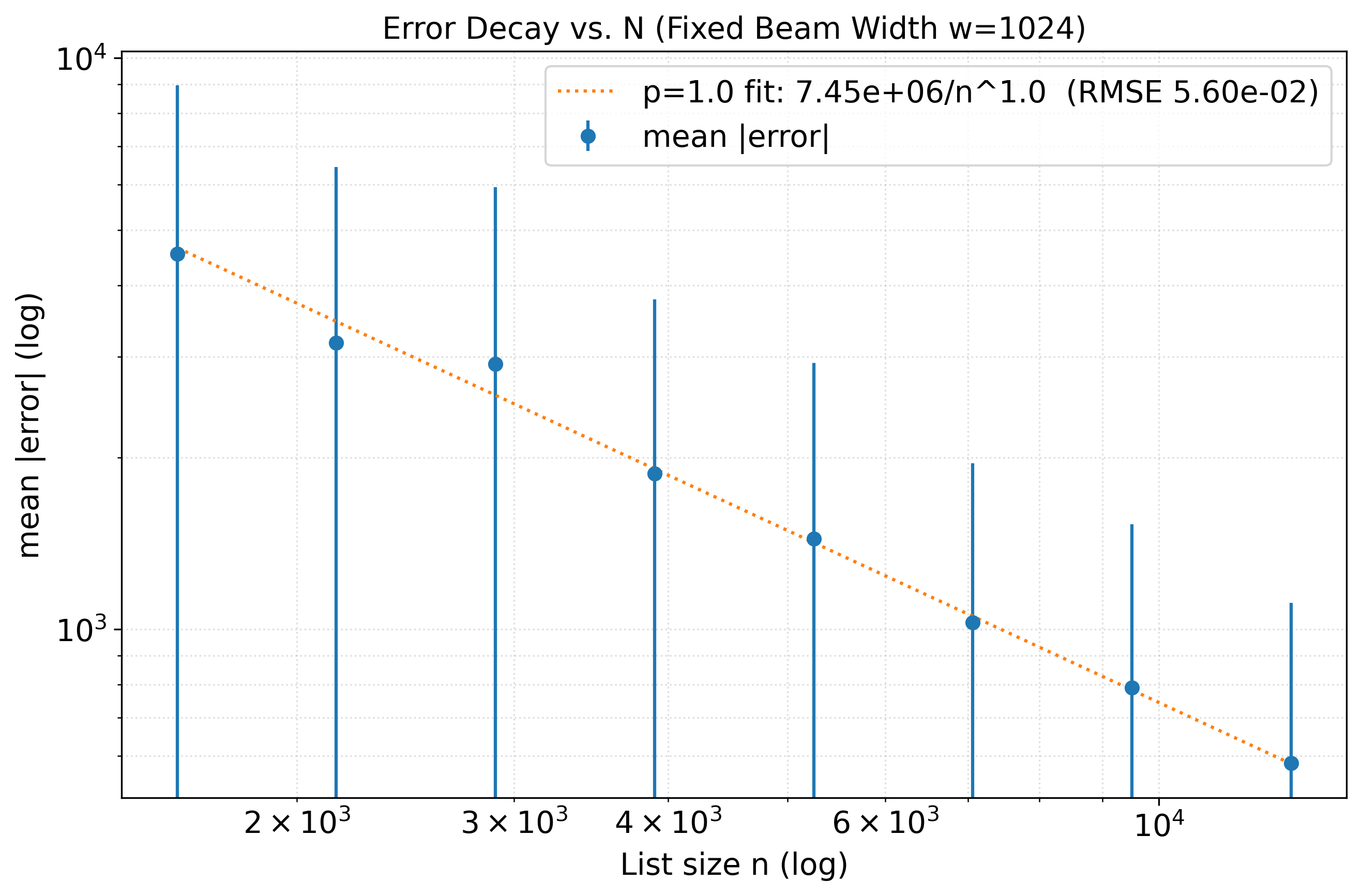}
    \caption{Effect of the number of elements ($n$) in the large-$n$ regime. The decay rate exhibits a clean $1/n$ dependency as burn-in costs become negligible.}
    \label{fig:large-n-dependency}
\end{figure}

\subsection{Ablation Studies}

Finally, we isolate the contributions of a key algorithmic component via an ablation; We replace the Phase~A bucketing rule with a simpler equi-sampling baseline.

\begin{figure}[htbp]
    \centering
    \includegraphics[width=1\linewidth]{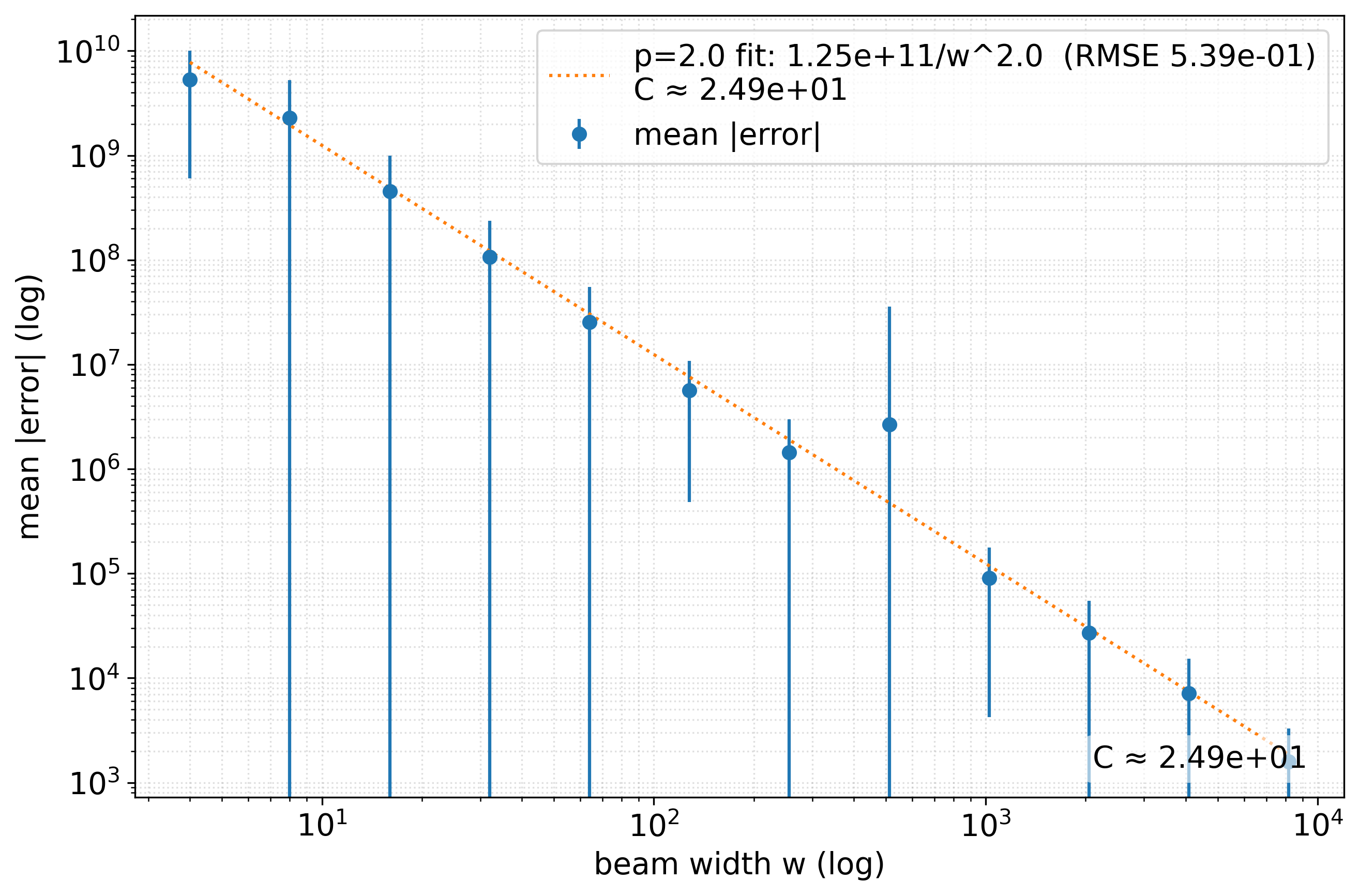}
    \caption{Ablation of Phase~A sampling: replacing the proposed bucketing rule with equi-sampling. The decay rate is preserved but the constant factor increases, consistent with reduced anchor coverage quality.}
    \label{fig:ablation_phaseA_equisample}
\end{figure}

\subsection{Tail Analysis}

This section evaluates the algorithm's performance when the target value lies in the extreme tails of the subset sum distribution.

\begin{figure}[htbp]
    \centering
    \includegraphics[width=1\linewidth]{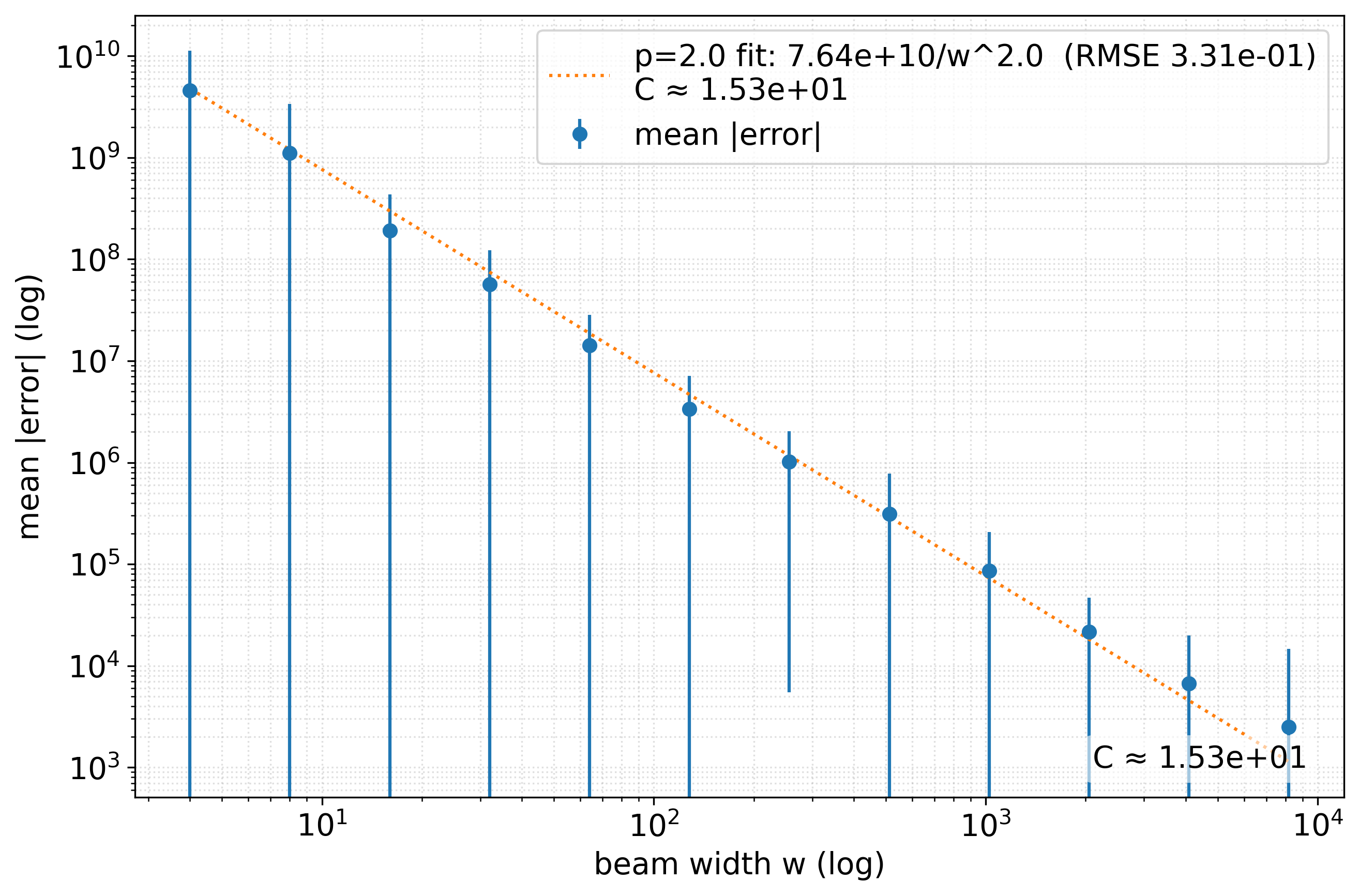}
    \caption{Performance scaling for a target moderately in the tail, $T=0.75\sum S$.}
    \label{fig:almost-tail}
\end{figure}

Figure \ref{fig:almost-tail} illustrates the algorithmic behavior when the target is set to $T=0.75\sum S$ (prior to symmetrization). In this regime, the inverse-quadratic error decay rate is preserved; however, the leading constant factor degrades significantly. Furthermore, systematic drift becomes more pronounced as the beam width $w$ increases. This drift occurs because the number of steps allocated to Phase B, defined as $n - n_{pre}$, exhibits an $O(\log w)$ dependency, which reduces the effective search depth.

\begin{figure}[htbp]
    \centering
    \includegraphics[width=1\linewidth]{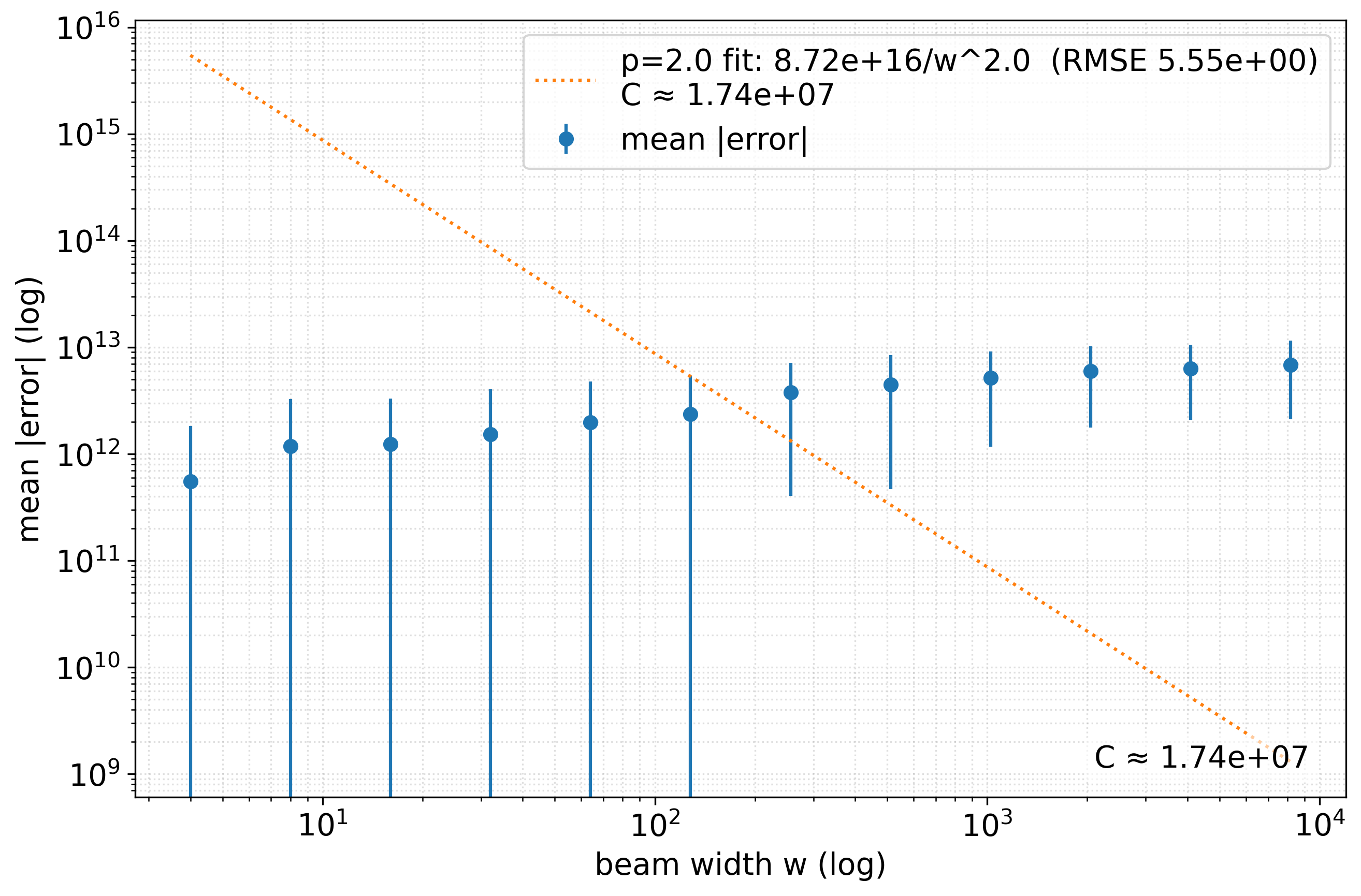}
    \caption{Failure mode for an extreme tail target, $T=0.95\sum S$.}
    \label{fig:in-tail}
\end{figure}

Conversely, Figure \ref{fig:in-tail} demonstrates the breakdown of the algorithm in an extreme tail scenario where $T=0.95\sum S$ (prior to symmetrization). Under these conditions, the required initial elements for the anchor phase and burn-in, $n_{pre}$, strictly exceed the total available elements $n$. The simulation's parameterization of $n=200$ is insufficient to absorb the $O(\log w)$ dependency inherent in $n_{pre}$. Consequently, Phase B is completely starved of elements, causing the inverse-quadratic error decay guarantee to fail. This empirical breakdown is consistent with the theoretical subexponential collapse predicted for extreme targets.

\subsection{Comparison Against Other Heuristics}

We additionally compare our algorithm against other heuristics in the literature, including an Arithmetic Optimization Algorithm (AOA) \cite{madugula_aoa_2022}, Particle Swarm Optimization (PSO) \cite{kennedy_particle_1995}, Tabu Search (Tabu)\cite{glover_tabu_1989, glover_tabu_1990}, Simulated Annealing (SA) \cite{kirkpatrick_optimization_1983}, the FPTAS by Gens and Levner \cite{gens_approximation_1979}, and a Genetic Algorithm (GA). As shown in Figure \ref{fig:heuristic_comparison_err_time}, the results show that the superior error scaling of our approach allows it to quickly surpass other approaches by multiple orders of magnitude \cite{nguyen_genetic_2004}.

To tune each hyperparameter, Optuna with 60 outer trials (amount of hyperparameter sets to try) and 3 inner trials (amount of trials to determine hyperparameter set effectiveness) per timeframe was used. We find that this is enough to ensure convergence across all timeframes within a factor of 2, which is negligible in logspace. Detailed hyperparameter suggestions for Optuna can be found in the attached code.

In this graph specifically, standard error is used rather than standard deviation, ensuring that the graph is readable.
\begin{figure}[htbp]
    \centering
    \includegraphics[width=\linewidth]{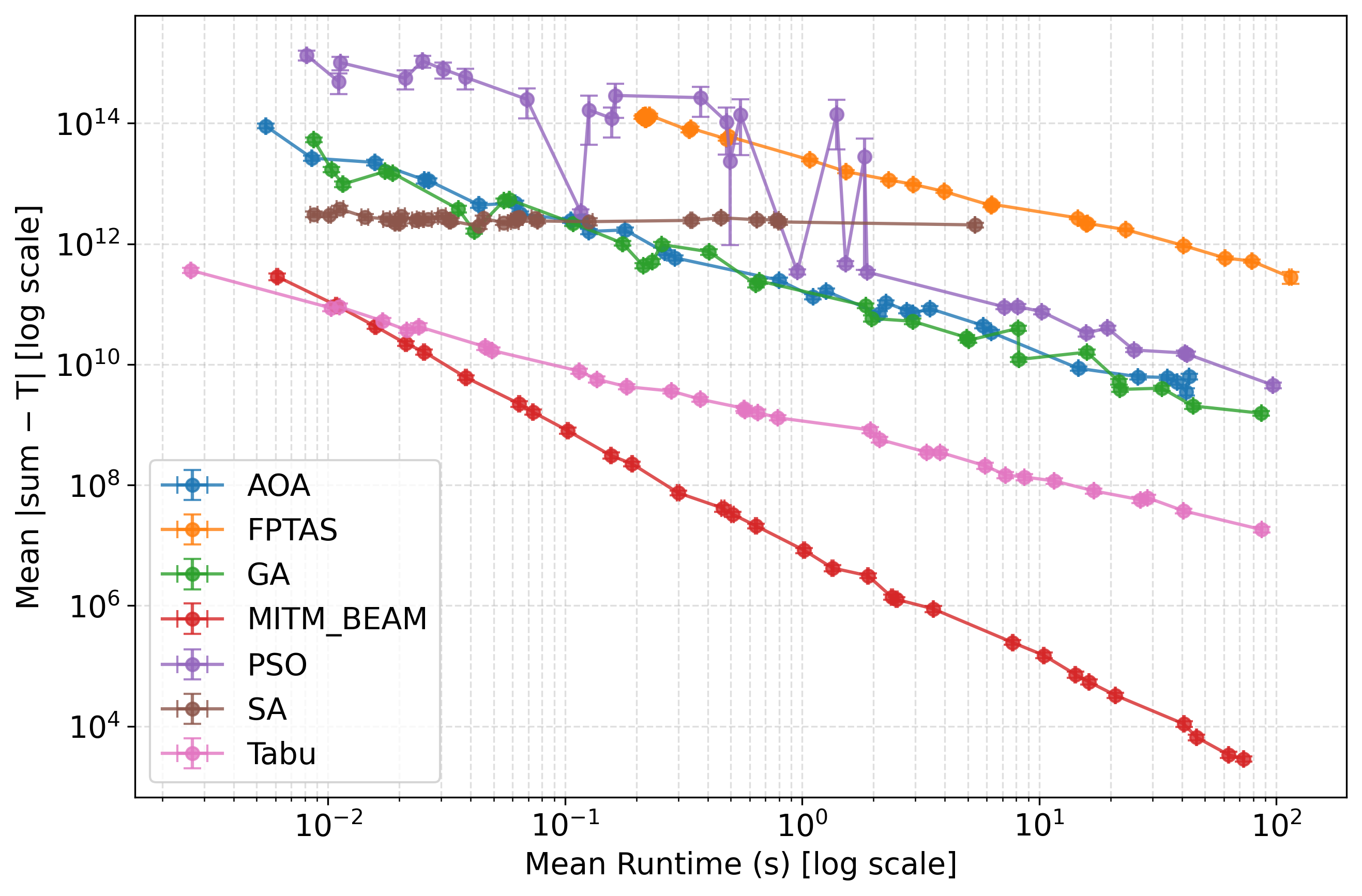}
    \caption{Approximation error versus wall--clock time for the proposed MITM beam search compared against standard heuristics: Arithmetic Optimization Algorithm (AOA), Particle Swarm Optimization (PSO), Tabu Search (Tabu), Simulated Annealing (SA), Fully Polynomial Time Approximation Scheme (FPTAS), and a Genetic Algorithm (GA). 
The beam--based method exhibits a markedly steeper error decay, quickly surpassing competing heuristics by several orders of magnitude in solution quality.}
    \label{fig:heuristic_comparison_err_time}
\end{figure}

\section{Conclusion}

In this work, we transitioned the Random Subset Sum Problem (RSSP) from a strictly exact-cryptographic or worst-case perspective into a robust expected-error framework. By introducing a Meet-In-The-Middle (MITM) beam search, we achieved a provable inverse-quadratic error decay of $\mathbb{E}\,|S^\star-T| = \Theta(B / (n w^2))$, while maintaining a running time of $\widetilde{O}(nw)$.

The theoretical foundation for this decay relies on operationalizing the mesh existence proofs of Da Cunha et al.~\cite{DaCunhaRSSPMesh}. While prior work established that an $O(B/w)$ mesh exists with high probability among $O(\log w)$ random elements, our Phase~A construction provides an explicit, constructive pathway to achieve this mesh without the $w^{O(C)}$ exponential state-space blowup characteristic of naive tree expansion. By applying structured bucketing and iterative trimming, we force the mesh to materialize in linearithmic time, creating a structure that Phase B can efficiently exploit. During its initial burn-in phase, Phase B relies on mechanics similar to Phase A. By modeling the Phase B iteration as an anchor convolution under a mean-field assumption, we derive a recursion that—when inverted and refined through dominance arguments—produces our stated scaling bound.

Crucially, our empirical evaluations reveal that this inverse-quadratic decay rate is remarkably robust to the underlying input distribution, extending well beyond our theoretical uniform-input assumptions. As demonstrated in Section 6, the core $O(w^{-2})$ scaling holds whether the inputs are strictly uniform, multimodal, approximately Gaussian, or even drawn from heavy-tailed Cauchy distributions with undefined variance. Furthermore, this robustness naturally extends to RSSP variants, including Vector Subset Sum and Bounded Taking Subset Sum (Appendix \ref{sec:variant-examples}). Beyond these theoretical guarantees, our framework is empirically fast, surpassing standard metaheuristics (such as AOA, PSO, and Genetic Algorithms) by multiple orders of magnitude in expected error. Together, these results position the MITM beam search as a highly practical baseline for subset sum approximation.

While this expected-error paradigm opens new avenues for studying the RSSP, important theoretical gaps remain. A primary challenge is establishing unconditional lower bounds on the expected error for any algorithm constrained to $\widetilde{O}(nw)$ time, which would formally define the limits of the small-gap regime. Finally, while we have shown that a greedy geometric decay achieves a $w^{-2}$ bound, a compelling direction for future research is determining whether the redundant encoding and alphabet-expansion strategies used in exact representation techniques can be adapted to fundamentally accelerate this expected error decay rate.

\begin{acks}
The authors would like to thank Portland State University for the compute power required to run simulations, Tucker Mastin for valuable feedback and insight regarding evaluation, and Maxwell Chen for helping with proof validation. Claude Code was used to generate experimental harnesses used in Section 6 and Appendix B, and to generate documentation for the supplemental material.

This material is supported by the National Science Foundation under Grant Nos. 2346732, 2318139, and 2019216.
\end{acks}

\bibliographystyle{ACM-Reference-Format}
\bibliography{ref}

\appendix
\section{Variance of Error}

\paragraph{Variance bound in the random-phase regime.}
Let $D_t$ denote the global gap at Phase~B step~$t$, and define the one-step improvement $\Delta_t := D_t - D_{t+1}\;\ge 0$.

In the random-phase, small-gap regime $D_t \ll \Delta=B/w$ we have
\[
\mathbb{E}[\Delta_t \mid D_t]
= \Theta\!\Bigl(\frac{w^2}{B}\Bigr) D_t^2,
\qquad
\operatorname{Var}(\Delta_t \mid D_t)
= \Theta\!\Bigl(\frac{w^2}{B}\Bigr) D_t^3.
\]
Set $Y_t := 1/D_t$.  A second-order expansion gives
\[
\Delta Y_t
= Y_{t+1}-Y_t
= \frac{\Delta_t}{D_t^2}
  + \frac{\Delta_t^2}{D_t^3}
  + o(\Delta_t^2).
\]

\paragraph{Drift of $Y_t$.}
Taking conditional expectations and using the above moments,
\[
\mathbb{E}[\Delta Y_t \mid \mathcal{F}_{t-1}]
= \Theta\!\Bigl(\frac{w^2}{B}\Bigr).
  \]
Summing over steps,
\[
\mathbb{E}[Y_T]
= Y_0 + \Theta\!\Bigl(\frac{w^2}{B}\,T\Bigr)
= \Theta\!\Bigl(\frac{w^2}{B}\,T\Bigr).
\]

\paragraph{Conditional variance of $\Delta Y_t$.}
From the expansion and $|\Delta_t|\le D_t$,
\[
\operatorname{Var}(\Delta Y_t \mid \mathcal{F}_{t-1})
= O\!\left(
\frac{\operatorname{Var}(\Delta_t \mid \mathcal{F}_{t-1})}{D_t^4}
+ \frac{\mathbb{E}[\Delta_t^4 \mid \mathcal{F}_{t-1}]}{D_t^6}
\right).
\]

Let $p_t=\Pr(\text{improve} \mid D_t)=O(\frac{w^2}{B}D_t)$ and conditional on improvement, the improvement size $S_t=D_t-D_{t+1}  \sim U(0, D_t)$. Then $\Delta_t=D_t-D_{t+1}=I_tS_t, \quad I_t \sim \text{Bernoulli}(p_t)$. Hence, for any $m \geq 1$, 
\[
\mathbb{E}[\Delta_t^m \mid D_t]
= \mathbb{E}[(I_t S_t)^m \mid D_t]
= \mathbb{E}[I_t S_t^m \mid D_t]
= p_t\,\mathbb{E}[S_t^m \mid D_t].
\]

Taking $m=4$ yields 
\[
\mathbb{E}[\Delta_t^4 \mid D_{t}] = p_t\mathbb{E}[S_t^4 \mid D_t]=O\!\bigl(\tfrac{w^2}{B}\bigr) D_t^5.
\]

Since
$\operatorname{Var}(\Delta_t \mid \mathcal{F}_{t-1})
   = \Theta\!\bigl(\tfrac{w^2}{B}\bigr) D_t^3$,
we obtain
\[
\operatorname{Var}(\Delta Y_t \mid \mathcal{F}_{t-1})
= O\!\Bigl(\frac{w^2}{B}\cdot\frac{1}{D_t}\Bigr).
\]

\paragraph{Variance of $Y_T$.}
By the martingale variance identity,
\[
\operatorname{Var}(Y_T)
= \sum_{t=1}^T
  \mathbb{E}\big[\operatorname{Var}(\Delta Y_t \mid \mathcal{F}_{t-1})\big]
= O\!\left(
\frac{w^2}{B}
\sum_{t=1}^T
\mathbb{E}\!\left[\frac{1}{D_t}\right]
\right).
\]
From the drift, $\mathbb{E}[1/D_t]=\mathbb{E}[Y_t]=
\Theta\!\bigl(\frac{w^2}{B}\,t\bigr)$, giving
\[
\operatorname{Var}(Y_T)
= O\!\left(
\frac{w^2}{B}
\sum_{t=1}^T
\frac{w^2}{B}\,t
\right)
= O\!\left(\frac{w^4}{B^2}\,T^2\right).
\]

\paragraph{Transforming back to $D_T$.}
Using the delta (or Lipschitz) bound for $g(y)=1/y$,
\[
\operatorname{Var}(D_T)
= \operatorname{Var}(g(Y_T))
\;\le\;
\frac{\operatorname{Var}(Y_T)}{(\mathbb{E}[Y_T])^4}
\]
\[
= O\!\left(
\frac{(w^4/B^2)T^2}{
(\tfrac{w^2}{B}T)^4}
\right)
= O\!\left(\frac{B^2}{w^4\,T^2}\right).
\]

\paragraph{Conclusion.}
Hence in the random-phase, small-gap regime,
\[
\boxed{\;
\operatorname{Var}(D_T)
= O\!\left(\frac{B^2}{w^4\,T^2}\right).
\;}  
\]
\section{Adapting to Variants} \label{sec:variant-examples}
\subsection{Bounded Taking}
\label{app:bounded-taking}

The Bounded Taking variant of Subset Sum restricts the solution to at most $k$ elements. Consider the canonical setting where the target is $T=0$, inputs are drawn from $U(-B, B)$, and a valid solution must contain at least one element. In this setting, the Bernoulli symmetrization transform is natively bypassed. We can adapt the MITM beam search to accommodate the cardinality constraint $k$ while preserving the expected error decay, provided $k$ is sufficiently large.

\subsubsection{Natural Logarithmic Sparsity and Error Decay}
Introducing a strict budget $k$ limits the maximum number of improvements the algorithm can make, as an improvement only occurs on an inclusion branch. Because the unconstrained expected error is $O(B/(n w^2))$, the algorithm requires a global distance reduction factor of roughly $nw$.

Surprisingly, the optimal solutions found by this framework are naturally logarithmically sparse.

\begin{lemma}[Concentration of Stochastic Error Decay]
\label{lem:stochastic-decay-gamma}
Let $D_0 = \Theta(B/w)$ be the initial Phase B gap. To achieve a target error of $D_{target} = \Theta(B/(nw^2))$ with probability $1-\delta$, Phase B requires a budget of only $k_B = \lceil \ln(nw) \rceil + O(\sqrt{\ln(nw)\ln(1/\delta)})$.
\end{lemma}

\begin{proof}
By Lemma \ref{lem:uniform-improvement}, conditional on an improvement, the new distance $D'$ is stochastically dominated by a uniform distribution on $[0, D]$. After $m$ successful improvements, the error can be upper-bounded by:
\[
D_m = D_0 \prod_{i=1}^m U_i, \quad U_i \stackrel{i.i.d.}{\sim} U(0, 1).
\]
To find the number of steps $m$ needed to achieve a reduction ratio of $D_0/D_m = nw$, we apply the negative natural logarithm:
\[
-\ln\left(\frac{D_m}{D_0}\right) = \sum_{i=1}^m -\ln(U_i).
\]
Let $X_i = -\ln(U_i)$. The standard uniform distribution transformed by $-\ln$ yields a standard exponential distribution: $X_i \sim \text{Exp}(1)$, with $\mathbb{E}[X_i]=1$ and $\operatorname{Var}(X_i)=1$. The sum of $m$ independent $\text{Exp}(1)$ variables exactly follows a Gamma distribution:
\[
S_m = \sum_{i=1}^m X_i \sim \Gamma(m, 1),
\]
where $\mathbb{E}[S_m] = m$ and $\operatorname{Var}(S_m) = m$.
We require $S_m \ge \ln(nw)$. Because $\mathbb{E}[S_{k_B}] = k_B$, setting $k_B \approx \ln(nw)$ achieves the target reduction in expectation. Applying standard Chernoff bounds for the Gamma distribution, padding the required budget by the standard deviation ensures concentration:
\[
k_B = \lceil \ln(nw) \rceil + c\sqrt{\ln(nw) \ln(1/\delta)},
\]
which guarantees $S_{k_B} \ge \ln(nw)$ with probability at least $1-\delta$.
\end{proof}

Phase A consumes at most $O(\log w)$ budget. Thus, if the global budget satisfies $k \ge \Theta(\log n + \log w)$, the cardinality constraint is essentially non-binding, since if $k < O(\log(nw))$ then the problem can be solved in subexponential time. The Bounded Taking variant will naturally converge on the full $O(B/(nw^2))$ unconstrained error decay, proving that our method creates solutions that are naturally sparse.

This provides a distinct advantage over existing approaches. Standard dynamic programming and FPTAS methods typically require an augmented state space to enforce cardinality constraints, scaling complexity by a factor of $k$. Furthermore, to the best of our knowledge, analyses of representation-based techniques~\cite{howgravegraham_joux_2010} and standard metaheuristics do not establish expected logarithmic sparsity bounds. Therefore, the ability of the MITM beam search to naturally converge on $O(\log(nw))$-sparse solutions represents a novel structural guarantee for the Bounded Taking variant.

\subsection{Adapting to Vector Subset Sum}
\label{app:vector-ssp}

The Vector Subset Sum Problem (VSSP) generalizes to $d$ dimensions:
given $\mathbf{v}_1, \dots, \mathbf{v}_n \in \mathbb{Z}^d$ with
$\|\mathbf{v}_i\|_\infty \le B$ and target $\mathbf{T} \in \mathbb{Z}^d$,
minimize $\|\sum_{\mathbf{v} \in V} \mathbf{v} - \mathbf{T}\|_2$.
The MITM beam search adapts with three changes: Phase~A buckets become
hypercube cells ($w^{1/d}$ per axis, $w$ total, side $\Delta = B/w^{1/d}$),
scoring uses Euclidean distance, and symmetrization applies coordinate-wise.

\paragraph{Phase~A.}
The logistic filling dynamics carry over: each filled cell contributes
$\Omega(1/w)$ filling probability per unfilled cell (via a $d$-dimensional
checkerboard disjointness argument), so all $w$ cells fill in $O(\log w)$
steps w.h.p., yielding mesh gap $O(B\sqrt{d}/w^{1/d})$.

\paragraph{Phase~B.}
The key difference is geometric. Each (beam element, anchor) pair at
distance $D$ contributes an improvement region that is a $d$-ball of
volume $V_d D^d$. With $\Theta(w^2)$ disjoint pairs in the small-gap
regime, the improvement probability becomes
$\Pr(\text{improve} \mid D) = \Theta(w^2 V_d D^d / B^d)$.
Conditional on improvement, the new distance satisfies
$\mathbb{E}[D'] = \frac{d}{d+1}D$ (expectation of the radial coordinate
of a uniform point in a $d$-ball). The resulting drift is
\[
\mathbb{E}[D_{t+1} - D_t \mid D_t = D]
\;\le\; -c_d\,\frac{w^2}{B^d}\,D^{d+1}.
\]
Substituting $Y_t = D_t^{-d}$ gives $Y_{t+1} \ge Y_t + dk$, so
$D_t = O(B/(w^2 t)^{1/d})$. Setting $t = \Theta(n)$:

\begin{theorem}[MITM beam for Vector SSP]
\label{thm:vector-mitm}
Under the same conditions as Theorem~\ref{thm:mitm-main} but in $d$
dimensions, the MITM beam search returns $\mathbf{S}^\star$ satisfying
\[
\mathbb{E}\!\left[\left\|\mathbf{S}^\star - \mathbf{T}\right\|_2\right]
= O\!\left(\frac{B}{(n\,w^2)^{1/d}}\right),
\]
in $O(nw\log w \cdot d)$ time and $O(wd)$ memory, using a KDTree as long as $n >> 2^d$.
\end{theorem}

As illustrated in Figure \ref{fig:vector-ssp}, this is empirically consistent across $d \in \{2, 3, 5\}$. The exponent degrades from $w^{-2}$ to $w^{-2/d}$, reflecting that $d$-ball volume scales as $D^{d}$, making each unit of distance reduction harder. 

\begin{figure}[htbp]
    \centering
    \includegraphics[width=1\linewidth]{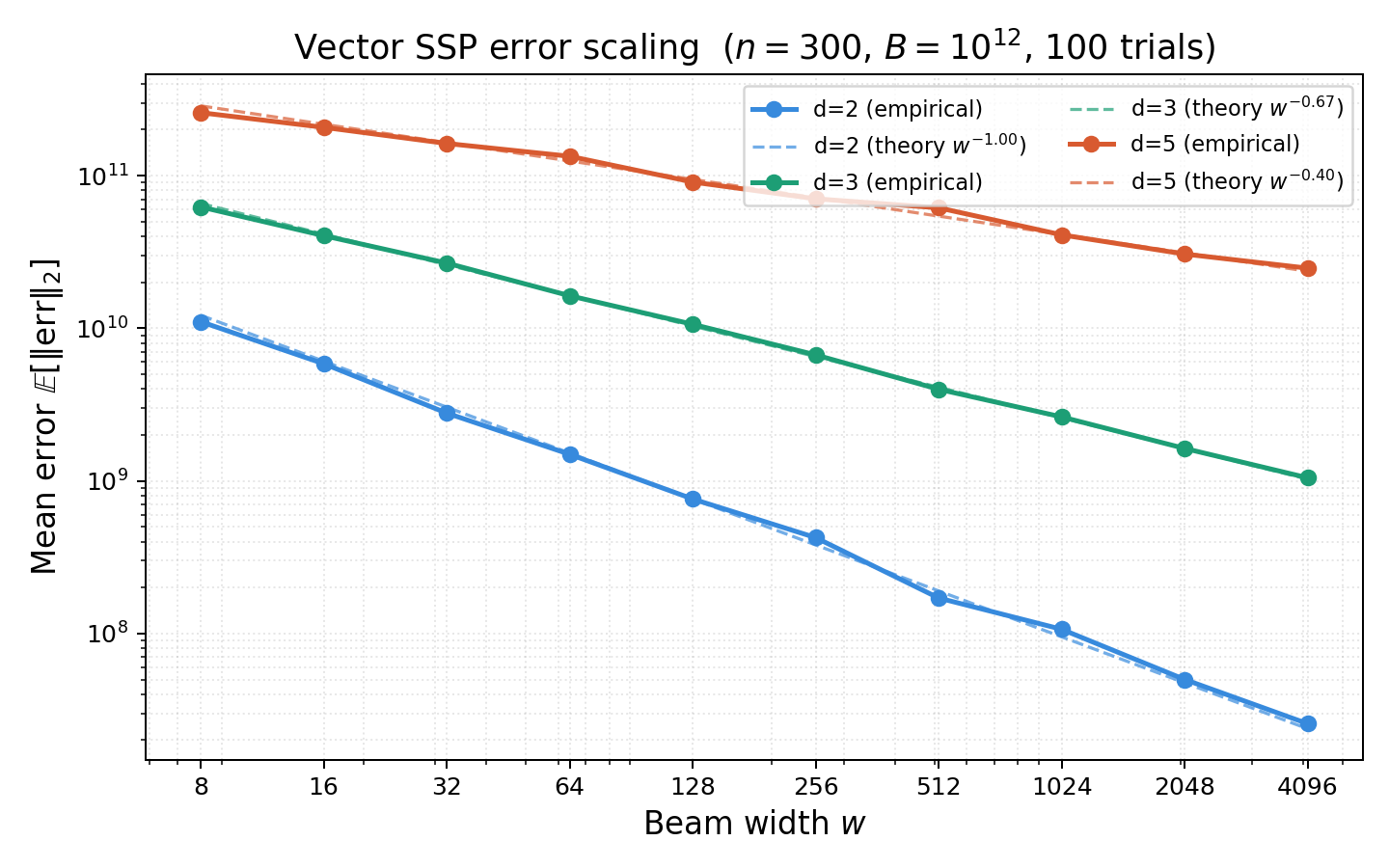}
    \caption{Empirical error scaling of the proposed MITM beam search applied to the Vector Subset Sum Problem across $d \in \{2, 3, 5\}$. Results confirm the theoretical $w^{-2/d}$ exponent degradation.}
    \label{fig:vector-ssp}
\end{figure}

\section{Proofs for Phase B Burn-in}
\label{app:burn-in}

This appendix provides the formal proofs for the burn-in dynamics summarized in Lemma~\ref{lem:burnin-summary}, specifically detailing the random walk into the target range and the monotonicity guarantees that allow Voronoi cells to fill.

\paragraph{Phase 1: reaching the anchor range.}
Without loss of generality assume all $z\in Z$ are positive. Before any beam element enters $Q$, the scoring rule preserves the largest beam values. Hence negative increments are never retained, and the maximum beam element evolves as
\[ M_k := \sum_{t=1}^k Y_t, \qquad Y_t := \max\{0,s_t\},\quad s_t\sim U([-B,B]). \]
The variables $\{Y_t\}$ are i.i.d.\ with exact mean and variance: $\mathbb{E}[Y_t]=\frac{B}{4}$, and $\operatorname{Var}(Y_t)=\frac{5B^2}{48}$. Let $m:=\min Z$. By the Central Limit Theorem, $M_k \sim \mathcal{N}\left(k\frac{B}{4},\; k\frac{5B^2}{48}\right)$. To ensure $M_k \ge m$ with high probability $1-\delta$, we solve for $k$ using the standard normal quantile $z_{1-\delta}$:
\[ k\frac{B}{4} - z_{1-\delta}\sqrt{k\frac{5B^2}{48}} \;\ge\; m. \]
This yields a required step count of $k \approx \frac{4m}{B} + O\!\left(\sqrt{\frac{m}{B}}\right)$.

\paragraph{Phase 2: expansion inside $Q$ (filling all Voronoi cells).}
Let $\{V_1,\dots,V_m\}$ be the Voronoi partition of the line induced by $Z$. Under the $\Theta(B/w)$-spacing guarantee from Phase A, there exists an absolute constant $\alpha\in(0,1)$ such that for all $j$, $\mathrm{Leb}(V_j \cap Q) \ge \alpha\cdot \frac{B}{w}$.

Define the set of occupied Voronoi cells at time $t$ as $\mathcal{F}_t := \bigl\{j : \exists\,x\in \mathcal{W}_t\cap Q\ \text{with}\ x\in V_j\bigr\}$, and let $K_t := |\mathcal{F}_t|$.

\begin{lemma}[Error Monotonicity]
\label{lem:err-monotone}
The minimum error/distance at time $t$, $D(t)$, monotonically decreases.
\end{lemma}
\begin{proof}
Let $\text{Dist}(\mathcal{S})=\min |\mathcal{S}_i-T|$. After an expansion $\mathcal{W'}_{i}=\mathcal{W}_i\cup\mathcal{W}_i+s_i$, $W'_{i}$ will be a superset of $W_i$. If the scoring heuristic $H$ satisfies $\text{Dist}(H(\mathcal{S}^{(1)}))\leq \text{Dist}(H(\mathcal{S}^{(2)}))$ where $\mathcal{S}^{(1)} \subseteq \mathcal{S}^{(2)}$, error monotonically decreases. The scoring heuristics trivially satisfy this property because preserving local minima also preserves global minima.
\end{proof}

\begin{lemma}[Filling Monotonicity] 
\label{lem:fill-monotone}
After any beam element $\mathcal{W}_i$ enters $Q$, any filled Voronoi cell cannot be vacated.
\end{lemma}
\begin{proof}
Inside $[-\min Z, \max Z]$, having one anchor every two buckets guarantees a spacing of at most $2\Delta$. The condition that $\mathcal{W}_i$ enters $Q$ ensures it does not stray more than $c\frac{B}{w} \le 2\Delta$. A minimum spacing of $2\Delta$ guarantees $\mathcal{H}\neq \emptyset$. By Lemma~\ref{lem:err-monotone}, minimum error is monotonically decreasing, hence the scoring heuristic permanently switches to $\text{OnePerBucket}$. A filled Voronoi cell is not subject to truncation that would vacate it.
\end{proof}

\begin{lemma}[Voronoi-cell filling after entering $Q$]
\label{lem:burnin-phase2-fill}
Assume that at time $t_0$ the beam has at least one element in $Q$. There exists an absolute constant $c>0$ such that for all $t\ge t_0$,
\[ \mathbb{E}[K_{t+1}-K_t \mid \mathcal{F}_t] \ge c\cdot (m-K_t)\cdot \frac{K_t}{w}. \]
Consequently,
\[ \Pr\Big(K_{t_0 + C_1\log w} = m\Big) \ge 1-w^{-\Omega(1)}. \]
\end{lemma}
\begin{proof}
Fix $t\ge t_0$ and condition on $\mathcal{F}_t$ with $K_t=K$. Pick any unoccupied cell index $j\notin \mathcal{F}_t$. For every $i\in\mathcal{F}_t$, there is $x_i\in \mathcal{W}_t\cap Q$ with $x_i\in V_i$. The event that $x_i+s_{t+1} \in V_j\cap Q$ requires $s_{t+1}\in (V_j\cap Q) - x_i$. Since $x_i\in Q$ and $Q$ has width $\Theta(B)$, the translate intersects $[-B,B]$ with measure $\Omega(\mathrm{Leb}(V_j\cap Q))$. There is an absolute constant $\beta>0$ such that $\Pr\big(x_i+s_{t+1}\in V_j\cap Q\ \big|\ \mathcal{F}_t\big) \ge \beta\cdot \frac{1}{w}$.

Taking the $K$ occupied representatives, a union bound gives $\Pr\big(\exists\, i\in\mathcal{F}_t:\ x_i+s_{t+1}\in V_j\cap Q\ \big|\ \mathcal{F}_t\big) \ge \min\!\left\{1,\ \beta\frac{K}{w}\right\}$. Summing over the $(m-K)$ unoccupied cells yields the discrete logistic growth lower bound. Standard comparison implies $K_t$ reaches $m$ in $O(\log w)$ additional steps.
\end{proof}
\section{Experimental Methodology and Parameters}
\label{app:experimental_methodology}

This appendix specifies the experimental protocol used to generate the
figures in Section~6.  The goal is reproducibility: the details below
give a summary of benchmark implementation used to generate all plots. Code for all experiments is provided in the supplementary material.

\subsection{Metric, averaging, and plots}
\label{app:metric_timing}

Each method run outputs a candidate subset sum $S$ and we report absolute
error
\[
\mathrm{err} \;=\; |S - T|.
\]
For each configuration (distribution, split rule, beam width), we run
independent trials and aggregate the sample mean and sample standard
deviation across trials. All reported scaling plots show $\mathbb{E}[\mathrm{err}]$
with error bars of one sample standard deviation. All scaling plots are
displayed on log--log axes.

\subsection{Instance generation}
\label{app:instances_generation}

Each trial samples $n$ i.i.d.\ integers from a specified family and a
target $T$ using one of two target rules.

\paragraph{Distribution families.}
We evaluate several qualitatively different i.i.d.\ input families.
Each family is evaluated in either a \emph{symmetric} form (approximately
supported on $[-B,B]$) or a \emph{nonnegative} form (supported on $[0,B]$).
All families are integer-valued via rounding.  Specifically:
\begin{itemize}[leftmargin=*]
  \item \textbf{Uniform.} Discrete uniform sampling on $[-B,B]$ (symmetric)
        or $[0,B]$ (nonnegative).
  \item \textbf{Normal.} Gaussian with mean $0$ and standard deviation
        $\sigma \approx B/3$, clipped to the relevant interval and rounded.
        The nonnegative form takes $|X|$ before clipping.
  \item \textbf{Lognormal.} A lognormal magnitude with moderate spread,
        rescaled so that typical values are comparable to $B$, then clipped
        and rounded.  The symmetric form assigns a random sign.  
  \item \textbf{Bimodal.} A two-component Gaussian mixture: in the symmetric
        setting the two modes are centered at approximately $\pm B/3$ with
        standard deviation about $B/10$; in the nonnegative setting the modes
        are centered near $0$ and $B$.
    \item \textbf{Student's $t$.} A Student's $t$ variate with $\nu$ degrees
        of freedom (default $\nu=2$, giving infinite variance), scaled by
        $B/4$ and clipped to $[-B,B]$ after rounding.  The nonneg\-ative
        form takes the absolute value before clipping.  Setting $\nu=1$
        yields the Cauchy distribution (infinite mean and variance).

\end{itemize}
All experiments use a fixed magnitude parameter $B$ across distributions
to keep dynamic range comparable (the benchmark uses $B=10^{12}$ unless
otherwise stated).

\paragraph{Target generation.}
Targets are generated as follows:
Choose a uniformly random subset of the sampled items (each item independently included with probability $1/2$) and set $T$ to the sum of that subset. This guarantees at least one exact solution exists.

\subsection{Proposed method configuration}
\label{app:proposed_method_config}

All experiments evaluate the meet-in-the-middle beam method described in
Section~5.  The item list is split into a left stage (Phase~A) and a right
stage (Phase~B).  Phase~A builds a set of anchors; Phase~B runs a width-$w$
beam search scored by distance to the set of residual targets induced by
the anchors.

\paragraph{Beam width grid.}
We evaluate a geometric grid of beam widths $w$ (powers-of-two style):
starting from a minimum value, repeatedly multiply by 2 up to a maximum.
Unless otherwise stated, this grid matches the one used in the benchmark
driver for the corresponding figure.

\paragraph{Split strategies.}
We evaluate several split rules:
\begin{itemize}[leftmargin=*]
  \item \textbf{Half split.} Split at $n/2$.
  \item \textbf{Fixed split.} Split at a fixed index (used to demonstrate
        failure when Phase~A is intentionally undersized).
  \item \textbf{Logarithmic split.} Split at $\lfloor c\log_2 w\rfloor$ for a
        chosen constant $c$, as motivated by the Phase~A analysis.
\end{itemize}

\paragraph{Phase~A anchoring rule.}
Phase~A expands the current anchor set by include/exclude of each left-half
item, then performs a width control step that retains a structured subset
of candidates.  We evaluate two Phase~A variants:
\begin{itemize}[leftmargin=*]
  \item \textbf{Bucketed random representative.} Partition the fixed domain
        $[-B/2,B/2]$ into $w$ equal-width buckets and keep one uniformly random
        candidate from each non-empty bucket.
  \item \textbf{Deterministic equi-sampling.} Sort the unique candidates and
        keep $w$ approximately evenly spaced representatives (max-spacing style).
\end{itemize}

\subsection{Parameter values}
\label{app:parameter-values}

Figures 1 through \ref{fig:heuristic_comparison_err_time} are run with 200 trials, $n=200$, $B=10^{12}$, with the target chosen as a random subset, and with the split point chosen as $\lfloor4\log(w)\rfloor$ unless otherwise stated. 

Figure \ref{fig:heuristic_comparison_err_time} is run with 100 trials, $n=300$, and $B=10^{15}$, with the target chosen as a random subset. 

Figure \ref{fig:vector-ssp} is run with 100 trials, $n=300$, and $B=10^{12}$.

Detailed hyperparameter breakdowns for each figure are available in the attached code.

\subsection{Fixed-exponent reference fits}
\label{app:aggregation_fits}

To visualize agreement with the predicted inverse-quadratic decay, we
overlay a fixed-exponent reference fit of the form
\[
\widehat{\mathbb{E}}[\mathrm{err}(w)] \approx \frac{c}{w^2}.
\]
The constant $c$ is estimated by least squares on the aggregated mean errors.
When reporting an implied constant in the theoretical scaling
$\mathbb{E}[\mathrm{err}] \approx C\cdot B/(n w^2)$, we convert via
\[
C \approx \frac{c\,n}{B}.
\]

\end{document}